\documentclass[12pt]{article}
\usepackage[round,sort,authoryear]{natbib}
\usepackage{cancel}
\usepackage{amsmath,amssymb,amsthm,amsfonts, fullpage, setspace}
\usepackage{eucal} 
\usepackage{subeqnarray,bm}
\usepackage[utf8]{inputenc}
\usepackage{graphicx}
\usepackage{epstopdf}
\usepackage{psfrag} 
\usepackage{color,float}
\usepackage{pdfpages}
\usepackage[colorlinks]{hyperref}
\hypersetup{colorlinks=true,citecolor=blue,pdfstartview=FitH, pdfpagemode=UseNone}
\usepackage{multirow}
\usepackage[bottom]{footmisc}
\usepackage{ifthen} 
\usepackage{subfig}
\usepackage{float}
\usepackage[english]{babel}
\usepackage[nottoc]{tocbibind}
\usepackage{calrsfs}
\usepackage{soul}
\usepackage{mwe}

\usepackage[capitalise,nameinlink]{cleveref}

\usepackage{booktabs}

\newtheorem{thm}{Theorem}[section]

\newtheorem{lemma}[thm]{Lemma}

\crefname{lemma}{Lemma}{Lemmas}
\Crefname{lemma}{Lemma}{Lemmas}

\newtheorem{proposition}[thm]{Proposition}

\newtheorem{remark}[thm]{Remark}

\newcommand{\ds}{\displaystyle}

\newcommand{\Z}{\mathbb{Z}}

\numberwithin{equation}{section}

\usepackage{calrsfs}

\newcommand{\W}{\textit{Wolbachia}}

\numberwithin{equation}{section}
\allowdisplaybreaks

\usepackage[table, dvipsnames]{xcolor}
\DeclareMathAlphabet{\pazocal}{OMS}{zplm}{m}{n}
\title{Travelling Waves in {\it Wolbachia} Spread Dynamics}
\author{Zhuolin Qu$^1$, Tong Wu$^1$, Eddy Kwessi$^2$\thanks{Corresponding author: ekwessi@trinity.edu}, \\
\small $^{1}$Department of Mathematics, University of Texas at San Antonio, San Antonio, TX, USA\\
\small $^{2}$Department of Mathematics,
Trinity University, San Antonio, TX, USA}
\date{}	

\begin{document}

\maketitle
\begin{abstract}
\W, a maternally transmitted endosymbiont, offers a powerful biological control strategy for mosquito-borne diseases such as dengue, Zika, and malaria. We develop an integro-difference equation (IDE) model that integrates \W's  nonlinear growth with spatially explicit mosquito dispersal kernels to study invasion dynamics in heterogeneous landscapes. Analytical results establish existence and uniqueness of monotone traveling waves and provide explicit estimates of invasion speeds as functions of dispersal and growth parameters. Four kernels: Gaussian, Laplace, exponential square-root, and Cauchy-represent a continuum from short- to long-range movement. Fat-tailed kernels generate faster, broader wavefronts, while compact ones limit spread. We also identify a critical bubble, the minimal localized profile required for sustained invasion. Numerical simulations in one- and two-dimensional domains confirm theoretical predictions and reveal parameter regimes governing invasion success. This framework quantifies how dispersal mechanisms shape \W's spread thus   informing targeted and efficient vector-control strategies.
\end{abstract}
\section{Introduction}
{\it Wolbachia}, a genus of maternally transmitted endosymbiotic bacteria, has received considerable attention in recent years due to its ability to manipulate host reproduction and suppress the transmission of vector-borne diseases such as dengue, Zika, and malaria.
When introduced into mosquito populations, \W~ can spread via cytoplasmic incompatibility (CI), a reproductive mechanism that gives infected females a selective advantage, leading to population replacement. This biological control strategy has been successfully deployed in field trials and holds promise for large-scale vector suppression programs.

Mathematical modeling has played a key role in understanding the population dynamics of \W. Classical models, including those by \cite{caspari1959} and \cite{turelli2014}, have provided insight into the conditions for invasion and persistence based on maternal inheritance and CI. However, many of these models often neglect spatial structure and assume homogeneous mixing, overlooking the critical role of dispersal and landscape heterogeneity in real-world settings.

Spatial dynamics are particularly important for assessing the success and stability of \W~ invasions. Environmental variability, habitat fragmentation, and localized releases all influence the spread of infection. Reaction-diffusion equations, as pioneered by \cite{fisher1937}, have been employed to study traveling wave phenomena in genetics and ecology. However, these models lack the specificity needed to capture {\it Wolbachia}-specific dynamics, such as fitness trade-offs and competition between infected and uninfected populations. Recent studies (e.g., \cite{barton2011,qu2022modeling}) have highlighted the importance of including spatially explicit dynamics to predict {\it Wolbachia} invasion under real-world conditions. 

A central object of study in such models is the traveling wave solution, which describes how an initially localized \textit{Wolbachia}-infected population invades and establishes itself in an uninfected population. The wavefront marks the transition zone between infected and uninfected regions, and its dynamics—such as shape and speed—are crucial for understanding how quickly and 
effectively \textit{Wolbachia} spreads~\cite{weinberger2002long, lewis2002spread}. In bistable systems, traveling waves connect two stable equilibria (typically the uninfected and fully infected states) and exhibit a sharp invasion threshold 
determined by the initial spatial distribution.

In this work, we analyze an integro-difference equation (IDE) model to describe the spatial spread of \W~ in mosquito populations. The model adopts a biologically grounded yet simplified \W~ growth function with a broad class of dispersal kernels. To ensure mathematical tractability, we model only the infection frequency using a scalar equation, rather than capturing the full mosquito life cycle. This reduced framework retains the essential bistable dynamics of \W, driven by CI, fitness costs, and maternal transmission. Our primary focus is on the spatial component: by incorporating a range of dispersal kernels, including Gaussian, Laplace, exponential square-root, and Cauchy, we explore how different movement behaviors shape traveling wave propagation and invasion thresholds. This approach enables rigorous analysis of wave existence, speed, and the critical bubble profile that separates successful invasion from failure.

The remainder of this paper is organized as follows: 
In \cref{sec:model}, we formulate the \W~ integro-difference equation (WIDE) model, including both the nonlinear growth function and a range of dispersal kernels, and conduct the fixed-point analysis. \Cref{TraWaves} presents theoretical results on traveling waves and the threshold condition for \W~ invasion. In \cref{sims}, we perform numerical simulations in both one-dimensional (1-D) and two-dimensional (2-D) radial geometries, comparing the effects of different kernels on wave speed and thresholds. \Cref{sec:discuss} concludes with a discussion of the biological implications and future directions for modeling \W-based vector control.

\section{Modelization} \label{sec:model}
We consider a discrete-time spatial model that captures the generation-to-generation dynamics of \W~ infection in mosquito populations. The model combines a nonlinear growth function representing CI and maternal inheritance with an integro-difference formulation that incorporates mosquito dispersal across space.

We begin by introducing the biological and mathematical structure of the \W~ growth model in a non-spatial setting. This forms the foundation for the spatial model presented in \cref{sec:space}, where we couple the growth dynamics with biologically motivated dispersal kernels. The descriptions of state variables, parameters, along with their baseline values, are provided in \cref{tab:baseline}.

\subsection{\textit{Wolbachia} Growth Model}
To establish the within-generation dynamics of 
\W, we consider a mixed mosquito population consisting of infected females $I_F(t)$, infected males $I_M(t)$,  uninfected females $U_F(t)$, and uninfected males $U_M(t)$ at generation $t$. Let $V_t\in[0, 1]$ denote the \W~ infection frequency, defined as the proportion of infected individuals in the population:
$$
\ds V_t=\frac{I_F(t)}{I_F(t)+U_F(t)}=\frac{I_M(t)}{I_M(t)+U_M(t)}\;.
$$
Let $s_f$ represent the relative fitness cost of infected females, and $1-s_f$ is the fitness of infected mosquitoes.
Let $s_h$ represent the CI intensity and $\mu$ be the maternal transmission leakage rate among all offspring of infected females. That is, $(1-\mu)\%$ of them are infected and $\mu\%$ of them are uninfected.
Given these assumptions, the infection frequency at the next generation, $V_{t+1}$, is derived as follows (see \cite{Yu2019}):
\begin{equation} \label{eq:growth}
V_{t+1}=f(V_t)=\frac{(1-\mu)(1-s_f)V_t}{s_hV_t^2-(s_h+s_f )V_t+1}\;, \quad  t\in \Z_+\;,
\end{equation}
This nonlinear model characterizes the {\it Wolbachia} generational spread dynamics in a well-mixed mosquito population. For $\mu=0$, the function exhibits bistability with three fixed points: $V_0^*=0$ (stable equilibrium, extinction), $V_1^*=s_f/s_h$ (unstable equilibrium), and $V_2^*=1$ (stable equilibrium, complete infection). For general $0\le\mu\le 1$, the fixed points are $V_0^*=0$,
\begin{equation} \label{eq:theta}  
V_1^*=\frac{ s_f + s_h - \sqrt{ s_f^2 - 4\mu s_h - 2 s_f s_h + 4\mu s_f s_h + s_h^2 } }{ 2 s_h }~~\text{(unstable)},
\end{equation}
and
\begin{equation}
\label{eq:V_mu}
V_2^*=\frac{ s_f + s_h + \sqrt{ s_f^2 - 4\mu s_h - 2 s_f s_h + 4\mu s_f s_h + s_h^2 } }{ 2 s_h } ~~\text{(stable)},
\end{equation}
where $0<V_1^*<V_2^*<1$. This bistability is a hallmark of the Allee effect (\cite{allee1949}), in which the survival or reproduction becomes inefficient when individuals are scarce, leading to a critical threshold (in this case $V_1^*$), below which populations decline. In \W~ dynamics, this manifests as an unstable equilibrium separating extinction from full infection: below the threshold, the infection fails to spread, while above it, invasion becomes self-sustaining.

\subsection{Spatial Model: Integro-Difference Formulation} \label{sec:space}
To incorporate spatial spread into the \W~ dynamics, we extend the scalar growth model into a spatially explicit setting using an IDE framework. Let $ V_t(x) $ denote the infection frequency at location  $x\in \mathbb{R}$ at the start of generation $t$. Following the classical IDEs (\cite{Kot1996}), the population undergoes two stages in each generation: (1) Local growth: at each location $x$, the infection frequency is updated by the local nonlinear growth function $ f(V_t(x)) $, as defined in \cref{eq:growth}, capturing the within-site dynamics due to CI and maternal transmission. (2) Dispersal: the infected population disperses across space according to a dispersal kernel $K(x,y)$, representing the probability density of movement from location $y$ to $x$. Assuming that the dispersal occurs after reproduction and that the dispersal kernel depends only on relative distance, that is $K(x,y) = K(x-y)$, we obtain the {\it Wolbachia}-integro-difference equation (WIDE):
\begin{equation}\label{eq:model}
V_{t+1}(x) = \int_{\mathbb R} K(x-y) f(x,V_t(y)) \mathrm{d}y.
\end{equation}
As shown by \cite{Kot1996}, when $f(V)\leq f(0)V_t$, integro-difference equations with such kernels may admit traveling-wave solutions. The shape of the dispersal kernel $K(x)$ influences both wave speed and spread effectiveness.  
In particular, we consider four biologically motivated kernel types: Gaussian kernel (short-range dispersal with rapid decay), Laplace kernel (exponential decay and allows moderate-range movement), Exponential square-root kernel (intermediate decay), and Cauchy kernel (fat-tailed with long-range dispersal). These kernels capture a range of potential mosquito movement behaviors and ecological scenarios. For example, in urban landscapes with barriers (e.g., buildings, roads), mosquito dispersal may be approximated by a Gaussian kernel, while open environments (e.g., coastal or agricultural regions) are better captured by fat-tailed distributions like the Cauchy kernel.
\begin{table}[H]
\centering \resizebox{1.15\textwidth}{!}{\begin{minipage}{1.3\textwidth}
\begin{tabular}{lll}
\toprule
& Description & Baseline\\  \midrule 
$V_t(x) $& Infection frequency at $t-$generation & --\\   
$f(V)$ & \W\, growth function & -- \\   
$K(x,y)$ & Spatial kernel & -- \\  \midrule 
$s_f$ & Relative fitness cost of infected females & 0.24 \\ 
$s_h$ & Cytoplasmic incompatibility intensity &  1\\ 
$\mu$ & Maternal transmission leakage & $0$ \\  \midrule  
$V_0^*$ & \textit{Wolbachia}-free fixed point & 0 \\  
$V_1^*$ & Unstable intermediate fixed point (Allele threshold)  &  $s_f/s_h = 0.24 $ \\  
$V_2^*$ & Stable high-infection fixed point & 1 \\  
$V_c^*(x)$ & ``Critical Bubble'': Nontrivial \W~threshold & -- \\  
\bottomrule
\end{tabular}
\caption{Model variables, parameters, and baseline values based on the wMel strain of \W~ in \textit{Aedes aegypti} mosquitoes \cite{walker2011wmel,hoffmann2014stability}. The dispersal kernel $K$ is given in \cref{tab:kernels}. }
\label{tab:baseline}
\end{minipage}}
\end{table}

\subsection{Stability of Fixed Points}
Before analyzing spatial propagation, it is useful to verify that the 
spatially homogeneous equilibria of the WIDE~\cref{eq:model} 
inherit the same stability characteristics as the fixed points of the local 
growth map~\cref{eq:growth}. 
This step ensures that spatial coupling through the dispersal kernel does not 
alter the bistable structure of the system, 
which underlies the formation of monotone traveling waves connecting the two 
stable equilibria.

\begin{thm}
Let \( V^* \) be a fixed point of the \textit{Wolbachia} growth function \( f(V) \).
Then:
\begin{enumerate}
\item[(i)] \( V^* \) is also a spatially homogeneous equilibrium of the WIDE.
\item[(ii)] Linearization of the WIDE around \( V^* \) yields the dispersion relation
\[
e^{\lambda(k)} = f'(V^*)\,\widehat{K}(k),
\]
where \( \widehat{K}(k) \) is the Fourier transform of the kernel \( K \).
Consequently:
\begin{itemize}
\item If \( |f'(V^*)| < 1 \), the equilibrium is {locally asymptotically stable (LAS)}.
\item If \( |f'(V^*)| > 1 \), the equilibrium is {unstable (UNS)}.
\item If \( |f'(V^*)| = 1 \), linearization is inconclusive.
\end{itemize}
\end{enumerate}
In particular, for the bistable growth function~\cref{eq:growth}, $V_0^* = 0 \;\text{(LAS, extinction)}$, $V_1^*$ (UNS, Allee threshold) and $V_2^*$ (LAS, complete infection) are given in ~\cref{eq:theta} and ~\cref{eq:V_mu}.
\end{thm}

\begin{proof}
\noindent $(i)$. As $V^*$ is the fixed point of $f(V)$, $f(V^*)=V^*$. Substituting into the WIDE \cref{eq:model}, 
\[
\int_{\mathbb R} K(x-y)\,f(V^*)\,dy
=\int_{\mathbb R} K(x-y)\,V^*\,dy
= V^*\int_{\mathbb R} K(x-y)\,dy = V^*.
\]
The last equality holds by the assumption that the kernel is normalized $\ds \int_\mathbb{R} K(z)~dz =1$. Thus, $V^*$ is invariant under the WIDE update and is a spatially homogeneous equilibrium of the WIDE.\\

\noindent $(ii)$. Let \(V_t(x) = V^* + \varepsilon(x,t)\) with \( |\varepsilon|\ll 1 \).
Substituting into the WIDE~\cref{eq:model} and linearizing gives
\[
\varepsilon(x,t+1)
= f'(V^*) \int_{\mathbb{R}} K(x-y)\, \varepsilon(y,t)\, dy.
\]
Using the Fourier ansatz 
\(\varepsilon(x,t) = e^{\lambda t} e^{ikx}\), we obtain
\[
e^{\lambda(t+1)} e^{ikx}
= e^{\lambda t} f'(V^*) e^{ikx} \widehat{K}(k),
\]
which simplifies to
\[
e^{\lambda} = f'(V^*)\,\widehat{K}(k).
\]
Solving for \( \lambda \), we obtain the condition:
\[
\lambda(k) = \log\left( f'(V^*)  \widehat{K}(k)  \right)\;.
\]
For stability, we need $\displaystyle \max_{k\in \mathbb{R}}~\lambda(k)<0$,
or equivalently,
\[
|f'(V^*)\,\widehat{K}(k)| < 1.
\]
Since most dispersal kernels satisfy \( 0<|\widehat{K}(k)|\le 1 \) with equality at \( k=0 \),
the spatial coupling does not change the local stability type of the 
homogeneous equilibria determined by \( f'(V^*) \).
\end{proof}

\noindent
These results confirm that the dispersal process preserves the bistable structure of the 
local dynamics.  
In the following sections, we analyze traveling-wave solutions that connect 
the two stable equilibria \(V_0^*\) and \(V_2^*\), 
separated by the unstable threshold \(V_1^*\).

\section{Traveling Waves}\label{TraWaves}

Traveling waves are a class of solutions in which the spatial profile propagates through the domain without changing shape at a constant speed. In the context of \W~ dynamics, such waves describe the invasion of \W-infected mosquitoes into uninfected ones over time. The spatial profile evolves into a wavefront that separates fully infected and uninfected regions. Mathematically, we define a traveling wave solution of the WIDE model as $V_t(x) = U(\xi) = U(x - ct),$ where $\xi = x - ct$ is the traveling wave coordinate, $U(\cdot)$ is the wave profile, and $c\in \mathbb{R}$ is the wave velocity. 

In this section, we establish the existence of monotone traveling wave solutions for the WIDE model under biologically relevant assumptions. We also provide an estimate of the associated wave velocity. In addition, due to the bistable behavior of the system, we present a threshold condition for reaching wave propagation from compactly supported initial conditions, a scenario that reflects the localized nature of field releases. 

\subsection{Existence of Traveling Waves}
We begin by stating the assumptions needed for the existence of monotone traveling wave solutions:
\begin{enumerate}
\item[(K):] The kernel $K\in L^1(\mathbb{R})$, $K(z)\ge0$, and normalized: $\ds \ds\int_\mathbb{R} K(z)~dz =1$\;.
\item[(F):] The function $f:[0,1]\to[0,1]$ is continuous, non-decreasing, and bistable with three fixed points: $f(0)=0$, $f(\theta)=\theta$, $f(1)=1$, for some $\theta\in(0,1)$, and 
$f(x)<x$ on $(0,\theta)$, $f(x)>x$ on $(\theta,1)$\;.
\item[(A):] Let $A=\{U(z): \mathbb R\rightarrow [0,1],~U \textit{~is~non-increasing~}\}$\;.
\end{enumerate}

We establish the existence result using a monotone iterative method, adapted to bistable IDEs. Several supporting lemmas (monotonicity, ordering, and tail behavior) are used to construct upper and lower solutions that converge to a monotone wave profile. We present the lemmas here; complete proofs are provided in Appendix \ref{sec:proof_lemma}. For convenience, and without loss of generality, we present the proof for the case \(V_2^* = 1\), so that the traveling wave connects \(1\) and \(0\); the same argument extends directly to the general case connecting \(V_2^*\) and \(0\).

\begin{lemma}[Monotonicity and Order Preservation]\label{lem:mono} 
Define the traveling wave operator:
$$
T_c(V)(z):=\int_{\mathbb{R}} K(z+c-y)f(V(y))dy.
$$
Then the following properties hold:
for any $U$, $V\in A$, $c$, $c_1$ and $c_2\in\mathbb{R}$,
\begin{enumerate}
\item[(i).] If $U\ge V$, then $T_c(U)\ge T_c(V)$.
\item[(ii).] If $U\in A$, then $T_c(U)\in A$. 
\item[(iii).] If $U\in A$ and $c_1\le c_2$, then $T_{c_1}(U)\ge T_{c_2}(U)$.
\end{enumerate}
\end{lemma}

\begin{lemma}[Tail limits]\label{lem:tails}
Let $W\in A$ with $\ds\lim_{z\to-\infty}W(z)=a$ and $\ds\lim_{z\to+\infty}W(z)=b$. Then
$\ds\lim_{z\to-\infty}T_0(W)(z)=f(a)$ and $\ds\lim_{z\to+\infty}T_0(W)(z)=f(b)$.
\end{lemma}

\begin{lemma}[Contact]\label{lem:contact}
\begin{enumerate}
\item[(i).] Let $U\in A$ with $\ds\lim_{z\rightarrow-\infty}U(z)=a\in(\theta,1)$ and $\ds\lim_{z\rightarrow\infty}U(z)=0$. Suppose the set $C=\{c: T_c(U)(z)\ge U(z),~\forall z\}$ is nonempty. Then there exists a finite $c^*=\sup C$, such that $$T_{c^*}(U)\ge U,$$ and there exists $z^*$ with $T_{c^*}(U)(z^*)=U(z^*)$. 
\item[(ii).] Let $U\in A$ with $\ds\lim_{z\rightarrow-\infty}U(z)=1$, and $\ds\lim_{z\rightarrow\infty}U(z)=a\in(0,\theta)$. Suppose the set $C=\{c: T_c(U)(z)\le U(z),~\forall z\}$ is nonempty. Then there exists a finite $c^*=\inf C$, such that $$T_{c^*}(U)\le U,$$ and there exists $z^*$ with $T_{c^*}(U)(z^*)=U(z^*)$. 
\item[(iii).] Let $U$, $V\in A$, $\ds1=\lim_{z\rightarrow-\infty}U(z)>\lim_{z\rightarrow-\infty}V(z)>\theta$, and $\ds\theta>\lim_{z\rightarrow\infty}U(z)>\lim_{z\rightarrow\infty}V(z)=0$, then the set $C=\{c: U(z+c)\ge V(z),~\forall z\}$ is nonempty, and there exists finite $c^*=\sup C$, such that $$U(z+c^*)\ge V(z).$$ Moreover, if $U$ is uniformly continuous, there exists $z^*$ with $U(z^*+c^*)=V(z^*)$. 
\end{enumerate}
\end{lemma}

Having established the key monotonicity, ordering, and contact properties of the
operator \(T_c\), we now proceed to demonstrate the existence of a nontrivial
traveling wave solution for the  WIDE~\cref{eq:model}. 
Specifically, we seek a monotone function (profile) \(\phi(z)\) that connects the two stable 
homogeneous equilibria of the system---the fully infected state (\(\phi(-\infty)=1\))
and the uninfected state (\(\phi(+\infty)=0\))---and propagates at a constant speed \(c\). The following theorem establishes both the existence of such a traveling wave and the
uniqueness (up to translation) of its wave speed.

\begin{thm}
There exists a non-increasing profile $\phi\in A$ and a unique speed $c$ such that $\phi=T_{c}(\phi)$, with $\phi(-\infty)=1$, $\phi(+\infty)=0$. Moreover, if $\widetilde\phi$ is another such profile, then there exists $z_0\in\mathbb R$ such that $\widetilde\phi(z)=\phi(z+z_0)$ for all $z$.
\end{thm}

\begin{proof}
Define the initial lower and upper profiles
\[
\phi_0^-(z) = (1 - \varepsilon)\, \mathbf{1}_{(-\infty,0]}(z),
\qquad
\phi_0^+(z) = \mathbf{1}_{(-\infty,0]}(z) + \varepsilon\, \mathbf{1}_{(0,\infty)}(z),
\]
where $1 - \varepsilon > \theta$ and $\varepsilon < \theta$. 
Clearly, $\phi_0^- < \phi_0^+$.

\paragraph{Step 1. Construction of the increasing sequences $\{\phi_n^-\}$ and $\{c_n^-\}$.}

~

\smallskip
\noindent
\textbf{Initialization.} 
Define the set
\[
C_1^- := \{\, c \in \mathbb{R} : T_c(\phi_0^-)(z) \ge \phi_0^-(z),~\forall z \in \mathbb{R} \,\}.
\]
By Lemma~\ref{lem:tails}, 
\[
\lim_{z\to -\infty} T_0(\phi_0^-)(z) = f(1-\varepsilon) > 1 - \varepsilon,
\]
and since $T_c(\phi_0^-)(z)$ is nonincreasing in $z$, there exists $c'$ such that 
$T_0(\phi_0^-)(z) > 1 - \varepsilon$ for all $z \le c'$. 
Hence,
\[
T_{c'}(\phi_0^-)(z)
= T_0(\phi_0^-)(z + c')
> 1 - \varepsilon
= \phi_0^-(z),
\quad \forall z \le 0.
\]
For $z > 0$, we have $\phi_0^-(z) = 0$ and $T_{c'}(\phi_0^-)(z) \ge 0 = \phi_0^-(z)$.
Therefore, $T_{c'}(\phi_0^-)(z) \ge \phi_0^-(z)$ for all $z$, implying that $C_1^-$ is nonempty.

By Lemma~\ref{lem:contact}, there exists a finite
\[
c_1^- := \sup C_1^-.
\]

Define $\phi_1^- := T_{c_1^-}(\phi_0^-)$. 
Then $\phi_1^- \ge \phi_0^-$, and there exists at least one point $z_1^-$ such that 
\(
\phi_1^-(z_1^-) = \phi_0^-(z_1^-).
\)

\smallskip
\noindent
\textbf{Iteration.}
Define
\[
C_2^- := \{\, c \in \mathbb{R} : T_c(\phi_1^-)(z) \ge \phi_1^-(z),~\forall z \in \mathbb{R} \,\}.
\]
Since $\phi_1^- = T_{c_1^-}(\phi_0^-)$ and $\phi_1^- \ge \phi_0^-$, Lemma~\ref{lem:mono} yields
\[
T_{c_1^-}(\phi_1^-) \ge T_{c_1^-}(\phi_0^-) = \phi_1^-,
\]
so $c_1^- \in C_2^-$. Hence $C_2^-$ is nonempty.
By Lemma~\ref{lem:contact}, there exists a finite
\[
c_2^- := \sup C_2^-\ge c_1^-.
\]
Set $\phi_2^- := T_{c_2^-}(\phi_1^-)$, so that $\phi_2^- \ge \phi_1^-$, 
and there exists at least one point $z_2^-$ satisfying 
\(\phi_2^-(z_2^-) = \phi_1^-(z_2^-)\).

\smallskip
\noindent
\textbf{Inductive step.}
Assume that for all $k \le n$,
\[
\phi_k^- = T_{c_k^-}(\phi_{k-1}^-), 
\quad c_k^- = \sup C_k^-,
\quad
C_k^- := \{\, c : T_c(\phi_{k-1}^-)(z) \ge \phi_{k-1}^-(z),~\forall z \,\},
\]
and that 
\(\phi_k^- \ge \phi_{k-1}^-\) and \(c_k^- \ge c_{k-1}^-\).

Define
\[
C_{n+1}^- := \{\, c \in \mathbb{R} : T_c(\phi_n^-)(z) \ge \phi_n^-(z),~\forall z \in \mathbb{R} \,\}.
\]
Since $\phi_n^- = T_{c_n^-}(\phi_{n-1}^-)$ and $\phi_n^- \ge \phi_{n-1}^-$, 
Lemma~\ref{lem:mono} gives
\[
T_{c_n^-}(\phi_n^-) \ge T_{c_n^-}(\phi_{n-1}^-) = \phi_n^-,
\]
so $c_n^- \in C_{n+1}^-$. Hence $C_{n+1}^-$ is nonempty. 
By Lemma~\ref{lem:contact}, there exists a finite
\[
c_{n+1}^- := \sup C_{n+1}^- \ge c_n^-.
\]
Define $\phi_{n+1}^- := T_{c_{n+1}^-}(\phi_n^-)$; then $\phi_{n+1}^- \ge \phi_n^-$, 
and there exists at least one $z_{n+1}^-$ such that 
\(\phi_{n+1}^-(z_{n+1}^-) = \phi_n^-(z_{n+1}^-)\).

\smallskip
\noindent
\textbf{Summary.}
Thus we have constructed sequences $\{\phi_k^-\}$ and $\{c_k^-\}$ starting from 
\[
\phi_0^-(z) = (1 - \varepsilon)\, \mathbf{1}_{(-\infty,0]}(z),
\]
such that for all $k \ge 0$,
\[
\phi_{k+1}^- = T_{c_{k+1}^-}(\phi_k^-), 
\qquad 
c_{k+1}^- \ge c_k^-,
\qquad 
\phi_{k+1}^-(z) \ge \phi_k^-(z),~\forall z.
\]
Moreover, for each $k \ge 0$, there exists $z_{k+1}^-$ with 
\(\phi_{k+1}^-(z_{k+1}^-) = \phi_k^-(z_{k+1}^-)\).

\paragraph{Step 2. Construction of the decreasing sequences $\{\phi_n^+\}$ and $\{c_n^+\}$.}

~

\smallskip
\noindent
\textbf{Initialization.}
Define
\[
C_1^+ := \{\, c \in \mathbb{R} : T_c(\phi_0^+)(z) \le \phi_0^+(z),~\forall z \in \mathbb{R} \,\}.
\]
By Lemma~\ref{lem:tails},
\[
\lim_{z\to \infty} T_0(\phi_0^+)(z) = f(\varepsilon) < \varepsilon,
\]
and since $T_c(\phi_0^+)(z)$ is nonincreasing in $z$, there exists $c'$ such that 
$T_0(\phi_0^+)(z) < \varepsilon$ for all $z \ge c'$. 
Hence,
\[
T_{c'}(\phi_0^+)(z)
= T_0(\phi_0^+)(z + c')
< \varepsilon
= \phi_0^+(z),
\quad \forall z \ge 0.
\]
For $z < 0$, since $\phi_0^+(z) = 1$ and $T_{c'}(\phi_0^+)(z) \le 1 = \phi_0^+(z)$,
we conclude that $T_{c'}(\phi_0^+)(z) \le \phi_0^+(z)$ for all $z$; hence $C_1^+$ is nonempty. By Lemma~\ref{lem:contact}, there exists a finite
\[
c_1^+ := \inf C_1^+.
\]
Define $\phi_1^+ := T_{c_1^+}(\phi_0^+)$. 
Then $\phi_1^+ \le \phi_0^+$, and there exists at least one point $z_1^+$ such that 
\(\phi_1^+(z_1^+) = \phi_0^+(z_1^+)\).

\smallskip
\noindent
\textbf{Iteration.}
Define
\[
C_2^+ := \{\, c \in \mathbb{R} : T_c(\phi_1^+)(z) \le \phi_1^+(z),~\forall z \in \mathbb{R} \,\}.
\]
Since $\phi_1^+ = T_{c_1^+}(\phi_0^+)$ and $\phi_1^+ \le \phi_0^+$, Lemma~\ref{lem:mono} implies
\[
T_{c_1^+}(\phi_1^+) \le T_{c_1^+}(\phi_0^+) = \phi_1^+.
\]
Hence $c_1^+ \in C_2^+$ and $C_2^+$ is nonempty.
By Lemma~\ref{lem:contact}, there exists a finite
\[
c_2^+ := \inf C_2^+ \le c_1^+.
\]
Set $\phi_2^+ := T_{c_2^+}(\phi_1^+)$, so $\phi_2^+ \le \phi_1^+$, and there exists at least one point $z_2^+$ satisfying 
\(\phi_2^+(z_2^+) = \phi_1^+(z_2^+)\).

\smallskip
\noindent
\textbf{Inductive step.}
Assume that for all $k \le n$,
\[
\phi_k^+ = T_{c_k^+}(\phi_{k-1}^+), 
\quad 
c_k^+ = \inf C_k^+,
\quad
C_k^+ := \{\, c : T_c(\phi_{k-1}^+)(z) \le \phi_{k-1}^+(z),~\forall z \,\},
\]
and that 
\(\phi_k^+ \le \phi_{k-1}^+\) and \(c_k^+ \le c_{k-1}^+\).
Define
\[
C_{n+1}^+ := \{\, c \in \mathbb{R} : T_c(\phi_n^+)(z) \le \phi_n^+(z),~\forall z \in \mathbb{R} \,\}.
\]
Since $\phi_n^+ = T_{c_n^+}(\phi_{n-1}^+)$ and $\phi_n^+ \le \phi_{n-1}^+$, Lemma~\ref{lem:mono} yields
\[
T_{c_n^+}(\phi_n^+) \le T_{c_n^+}(\phi_{n-1}^+) = \phi_n^+.
\]
Thus $c_n^+ \in C_{n+1}^+$ and $C_{n+1}^+$ is nonempty. 
By Lemma~\ref{lem:contact}, there exists a finite
\[
c_{n+1}^+ := \inf C_{n+1}^+\le c_n^+.
\]
Define $\phi_{n+1}^+ := T_{c_{n+1}^+}(\phi_n^+)$; then $\phi_{n+1}^+ \le \phi_n^+$,
and there exists at least one $z_{n+1}^+$ such that 
\(\phi_{n+1}^+(z_{n+1}^+) = \phi_n^+(z_{n+1}^+)\).

\smallskip
\noindent
\textbf{Summary.}
We have constructed monotone sequences $\{\phi_k^+\}$ and $\{c_k^+\}$ starting from
\[
\phi_0^+(z) = \mathbf{1}_{(-\infty,0]}(z) + \varepsilon\, \mathbf{1}_{(0,\infty)}(z),
\]
such that for all $k \ge 0$,
\[
\phi_{k+1}^+ = T_{c_{k+1}^+}(\phi_k^+), 
\qquad 
c_{k+1}^+ \le c_k^+,
\qquad 
\phi_{k+1}^+(z) \le \phi_k^+(z),~\forall z.
\]
Moreover, for each $k \ge 0$, there exists $z_{k+1}^+$ with 
\(\phi_{k+1}^+(z_{k+1}^+) = \phi_k^+(z_{k+1}^+)\).

\paragraph{Step 3. Verification that $c_k^- \le c_k^+$ for all $k$.}~

\smallskip
\noindent
\textbf{Base case $k=1$.}
Since $\phi_0^- < \phi_0^+$, Lemma~\ref{lem:mono} implies
\[
T_c(\phi_0^-) < T_c(\phi_0^+), \qquad \forall c \in \mathbb{R}.
\]
Assume, for contradiction, that $c_1^- > c_1^+$. Then
\[
\phi_1^- 
= T_{c_1^-}(\phi_0^-) 
\le T_{c_1^-}(\phi_0^+) 
\le T_{c_1^+}(\phi_0^+) 
= \phi_1^+.
\]
However, from the constructions in Steps 1–2,
\[
\phi_1^- \ge \phi_0^-, 
\qquad 
\phi_1^+ \le \phi_0^+,
\]
so that
\[
\phi_0^-(0) = 1 - \varepsilon 
\quad \Rightarrow \quad 
\phi_1^-(0) \ge 1 - \varepsilon,
\qquad
\phi_0^+(0^+) = \varepsilon 
\quad \Rightarrow \quad 
\phi_1^+(0) \le \varepsilon.
\]
Since $1 - \varepsilon > \varepsilon$, this yields 
\(\phi_1^-(z) > \phi_1^+(z)\) for $z$ near $0$, contradicting $\phi_1^- \le \phi_1^+$.  
Hence \(c_1^- \le c_1^+\).

\smallskip
\noindent
\textbf{Inductive step.}
Suppose $c_k^- \le c_k^+$ for some $k \ge 1$, and consider $c_{k+1}^-$ and $c_{k+1}^+$. Let \( f^{k} \) denote the \(k\)-fold composition of \(f\), i.e., \( f^{k}(x)=f(f(\cdots f(x)\cdots)) \) applied \(k\) times. 
By Lemma~\ref{lem:tails},
\[
\lim_{z\to -\infty} \phi_k^-(z) = f^{k}(1 - \varepsilon) > 1 - \varepsilon,
\qquad
\lim_{z\to \infty} \phi_k^-(z) = f^{k}(0) = 0,
\]
and
\[
\lim_{z\to -\infty} \phi_k^+(z) = f^{k}(1) = 1,
\qquad
\lim_{z\to \infty} \phi_k^+(z) = f^{k}(\varepsilon) < \varepsilon.
\]
Applying Lemma~\ref{lem:contact} to $\phi_k^-$ and $\phi_k^+$, we find a finite $c^*$ such that 
\[
\phi_k^+(z + c^*) \ge \phi_k^-(z), \qquad \forall z,
\]
and there exists $z^*$ with 
\(
\phi_k^-(z^*) = \phi_k^+(z^* + c^*).
\)
Assume, for contradiction, that $c_{k+1}^- > c_{k+1}^+$.  
Using $\phi_k^+(z + c^*) \ge \phi_k^-(z)$ and Lemma~\ref{lem:mono}, we obtain
\[
\phi_{k+1}^-(z)
= T_{c_{k+1}^-}(\phi_k^-)(z)
\le T_{c_{k+1}^-}(\phi_k^+)(z + c^*)
< T_{c_{k+1}^+}(\phi_k^+)(z + c^*)
= \phi_{k+1}^+(z + c^*).
\]
Evaluating at $z = z^*$ yields
\[
\phi_{k+1}^-(z^*) < \phi_{k+1}^+(z^* + c^*).
\]
Yet, by construction,
\[
\phi_{k+1}^-(z^*) \ge \phi_k^-(z^*) 
= \phi_k^+(z^* + c^*) 
\ge \phi_{k+1}^+(z^* + c^*),
\]
a contradiction.  
Therefore \(c_{k+1}^- \le c_{k+1}^+\).

\smallskip
\noindent
By induction, $c_k^- \le c_k^+$ for all $k \ge 1$.

\paragraph{Step 4. Existence of a traveling wave with a unique speed.}

\smallskip
\noindent
By the monotone convergence theorem, 
Since each sequence is monotone and bounded 
(\(\phi_k^- \uparrow \le 1\), \(c_k^- \uparrow \le c_1^+\); 
\(\phi_k^+ \downarrow \ge 0\), \(c_k^+ \downarrow \ge c_1^-\)),
there exist limits
\[
\phi_k^- \longrightarrow \phi^-,
\qquad
c_k^- \longrightarrow c^-,
\qquad
\phi_k^+ \longrightarrow \phi^+,
\qquad
c_k^+ \longrightarrow c^+,
\]
with $c^- \le c^+$. 
Moreover,
\[
T_{c^-}(\phi^-) = \phi^-,
\qquad
T_{c^+}(\phi^+) = \phi^+.
\]

\smallskip
\noindent
\textbf{Equality of the limits.}
We next show that $c^- = c^+$ and that $\phi^-$ and $\phi^+$ coincide up to a spatial shift. For any $k > 1$, by Lemma~\ref{lem:tails},
\[
\lim_{z\to -\infty} \phi_k^-(z) 
= f^{k}(1-\varepsilon)
< 1
= \lim_{z\to -\infty} \phi_k^+(z),
\qquad
\lim_{z\to \infty} \phi_k^-(z) = 0 
< f^{k}(\varepsilon)
= \lim_{z\to \infty} \phi_k^+(z).
\]
Hence, in the limit,
\[
\lim_{z\to -\infty} \phi^-(z)
\le \lim_{z\to -\infty} \phi^+(z) = 1,
\qquad
\lim_{z\to \infty} \phi^-(z)
= 0
\le \lim_{z\to \infty} \phi^+(z).
\]
By Lemma~\ref{lem:contact}, there exists $z_0 \in \mathbb{R}$ such that
\[
z_0 = \sup \bigl\{\, z' : \phi^+(z + z') \ge \phi^-(z),~\forall z \in \mathbb{R} \,\bigr\},
\]
so that $\phi^-(z) \le \phi^+(z + z_0)$ for all $z$, and equality holds at some point:
\[
\exists\, z^* \in \mathbb{R} \quad \text{such that} \quad
\phi^-(z^*) = \phi^+(z^* + z_0).
\]

\smallskip
\noindent
\textbf{Uniqueness of the speed.}
Suppose, to the contrary, that is, without loss of generality $c^- < c^+$. 
Let $\Delta = c^+ - c^- > 0$, and define
\[
\psi(z) := \phi^+(z + z_0).
\]
Then $\psi(z) \ge \phi^-(z)$ for all $z$, and $\psi(z^*) = \phi^-(z^*)$.
From the fixed-point relations,
\[
\phi^-(z) = T_{c^-}(\phi^-)(z)
= \int_{\mathbb{R}} K(z + c^- - y)\, f(\phi^-(y))\, dy,
\]
and
\[
\psi(z) = T_{c^+}(\psi)(z)
= \int_{\mathbb{R}} K(z + c^+ - y)\, f(\psi(y))\, dy.
\]
Using the change of variable $t = y + \Delta$, we can rewrite
\[
\psi(z)
= \int_{\mathbb{R}} K(z + c^- - t)\, f(\psi(t - \Delta))\, dt.
\]
Consequently,
\[
\psi(z^*) - \phi^-(z^*)
= \int_{\mathbb{R}} K(z^* + c^- - t)
\bigl(f(\psi(t - \Delta)) - f(\phi^-(t))\bigr)\, dt.
\]
Since $K > 0$, $f$ is nondecreasing, and $\psi \not\equiv \phi^-$ with $\Delta \ne 0$,
the integrand on the right-hand side is nonnegative and positive on a set of positive measure, which contradicts
$\psi(z^*) = \phi^-(z^*)$.
Therefore, $c^- = c^+$.

\smallskip
\noindent
It then follows that $\phi^- = \phi^+(\cdot + z_0)$,
yielding a single nonincreasing profile $\phi \in A$ and a unique speed $c$ satisfying
\[
T_c(\phi) = \phi,
\qquad 
\lim_{z\to -\infty} \phi(z) = 1,
\quad
\lim_{z\to \infty} \phi(z) = 0.
\] 
\end{proof}

\subsection{Traveling Wave Speed}
Having established the existence of monotone traveling waves connecting the stable equilibria of the WIDE, we now derive an explicit integral expression for the wave propagation speed.  This relation links the nonlinear growth term $f(v)$, the spatial profile $\phi$, and the dispersal kernel $K$.  
For convenience, we focus on the case $\mu=0$, where the upper equilibrium satisfies $V_2^* = 1$. The same derivation extends directly to the general case with $\mu>0$ by replacing the upper state~$1$ with $V_2^*$.

\begin{thm}
Consider the WIDE model \cref{eq:model} with the {\it Wolbachia} growth function 
$f(v)$ given in \cref{eq:growth}, and assume $\mu = 0$. 
Let $\phi(x-ct)$ be a smooth monotone traveling wave profile satisfying 
$\ds \lim_{x\rightarrow-\infty}\phi(x)=1$, $\ds \lim_{x\rightarrow\infty}\phi(x)=0$, and $\phi'(x)<0$. 
Assume in addition that
\[
f \text{ is Lipschitz on }[0,1],\qquad \phi' \in L^1(\mathbb{R}),\qquad \int_{\mathbb{R}}|z|\,K(z)\,dz<\infty.
\]
and 
$$\lim_{x\to\pm\infty}\phi^{(n)}(x)=0 \quad \text{for all } n\ge1.$$
Then the propagation speed $c$ is given by
\begin{equation}\label{eq:speed}
c = \int_0^1 \frac{v-f(v)}{\phi'\!\big(\phi^{-1}(v)\big)}\,dv.
\end{equation}
\end{thm}

\begin{proof}
Substituting the traveling--wave form $V_t(x) = \phi(x-ct)$ into 
the WIDE~\cref{eq:model} yields the steady profile equation
\begin{equation}\label{eq:profile_eq}
    \phi(x-c) 
    = \int_{\mathbb R} K(x-y)\, f(\phi(y))\, dy.
\end{equation}
Since $\phi$ is smooth, we expand $\phi(x-c)$ in a Taylor series around $x$:
\[
\phi(x-c) 
    = \phi(x) - c\phi'(x) + \sum_{n=2}^{\infty}\frac{(-c)^n}{n!}\phi^{(n)}(x).
\]
Substituting this into \cref{eq:profile_eq} and simplifying gives
\[
- c \phi'(x)
+ \sum_{n=2}^{\infty}\frac{(-c)^n}{n!}\phi^{(n)}(x)
= \int_{\mathbb R} K(x-y)\, f(\phi(y))\, dy - \phi(x).
\]
Integrating over $\mathbb{R}$ and using the boundary conditions 
$\ds \lim_{x\rightarrow-\infty}\phi(x)=1$ and $\ds \lim_{x\rightarrow\infty}\phi(x)=0$, we obtain
\[
\int_{\mathbb R} -c\phi'(x)\, dx = 
    -c\,[\phi(+\infty)-\phi(-\infty)] = c.
\]
The higher derivative terms vanish upon integration since 
$\phi^{(n)}\rightarrow0$ as $x\rightarrow\pm\infty$ for $n\ge1$. 
Thus,
\[
c = \int_{\mathbb R}
    \left(\int_{\mathbb R} K(x-y)\, f(\phi(y))\, dy - \phi(x)\right) dx.
\]
To simplify this expression, we note that the difference
\[
\int_{\mathbb R}
    \left(\int_{\mathbb R} K(x-y)\, f(\phi(y))\, dy - f(\phi(x))\right) dx
\]
is zero whenever $K$ is normalized, $f$ is Lipschitz, $\phi' \in L^1$, 
and $\int |z|K(z)\,dz < \infty$.  
Indeed, exchanging the order of integration gives
\[
\iint_{_{\mathbb R^2}} (K(x-y)-\delta(x-y)) f(\phi(y))\, dx\, dy
    = \int_{\mathbb R} (1-1)f(\phi(y))\, dy = 0,
\]
where $\delta(x)$ is the Dirac delta function. Hence
\[
c = \int_{\mathbb R} (f(\phi(x)) - \phi(x))\, dx.
\]
Since $\phi$ is smooth, strictly decreasing ($\phi'<0$), and satisfies $\ds \lim_{x\rightarrow-\infty}\phi(x)=1$ and $\ds \lim_{x\rightarrow\infty}\phi(x)=0$, 
it defines a $C^1$ bijection $\phi:\mathbb{R}\to(0,1)$ with inverse $x=\phi^{-1}(v)$. 
Using the change of variable $v=\phi(x)$, $dv=\phi'(x)\,dx$, we obtain the exact identity
\[
c = \int_0^1 \frac{v-f(v)}{\phi'\!\big(\phi^{-1}(v)\big)}\,dv.
\]
\end{proof}

\begin{remark}[Direction criterion]
Let $s(v):=1/\phi'(\phi^{-1}(v))<0$, so that
\[
c \;=\; \int_0^1 s(v)\,(v-f(v))\,dv.
\]
Let $S\subset(0,1)$ denote the transition region where $|v-f(v)|$ is concentrated, and assume $s$ varies slowly on $S$ while $|v-f(v)|$ is small off $S$. Let $\ds\bar{s}=\frac{1}{|S|}\int_Ss(v)\,dv$. Then
\[
c \;=\; \bar{s} \!\int_0^1 (v-f(v))\,dv \;+\; \int_S \!\big(s(v)-\bar{s}\big)(v-f(v))\,dv
      \;+\; \int_{(0,1)\setminus S} \!\big(s(v)-\bar{s}\big)(v-f(v))\,dv,
\]
and
\begin{equation*}
\big|\,c - \bar{s} \!\int_0^1 (v-f(v))\,dv\big|
\;\le\; \sup_{v\in S}\!|s(v)-\bar{s}|\,
        \int_S |v-f(v)|dv
      \;+\; \big(|\bar{s}|+\|s\|_{L^\infty}\big)\!\int_{(0,1)\setminus S}\! |v-f(v)|dv.
\end{equation*}
In particular, if $s$ is nearly constant on $S$ and $|v-f(v)|$ is negligible off $S$, then
\[
c \;\approx\; \bar{s} \!\int_0^1 (v-f(v))\,dv.
\]
Hence, the sign of $c$ is primarily determined by the sign of the area 
$\int_0^1 (v-f(v))\,dv$, with opposite sign to $\bar s$.  
Specifically,
\begin{equation}\label{eq:sign}
\operatorname{sign}(c)\approx 
-\operatorname{sign}\!\left(\int_0^1 (v-f(v))\,dv\right).    
\end{equation}
When the positive and negative regions of $(v-f(v))$ balance exactly,
$\int_0^1 (v-f(v))\,dv=0$, this approximation gives $c\approx0$, 
corresponding to the \emph{equal–area condition} often observed 
numerically at the transition between left– and right–moving waves.  
This criterion is consistent with the direction–selection rule in 
bistable reaction–diffusion systems and agrees with the numerical 
trends reported in \cref{sec:wavespeed}.
\end{remark}


\begin{remark}
If we consider the general case $0\le \mu \le 1$, then the wave speed estimate is given by
$$c=\int_0^{V_2^*} \frac{v-f(v)}{\phi'\!\big(\phi^{-1}(v)\big)}\,dv,$$
where $V_2^*$ is the stable fixed point, given in \cref{eq:V_mu}. Note that $V_2^*$ replaces the upper limit in ~\cref{eq:speed}. 
\end{remark}

\subsection{Critical Bubble: Threshold for Establishment of Traveling Waves given Compact Support}\label{sec:bubble}

In models with Allee effects, there exists a threshold density below which \textit{Wolbachia} cannot establish because an insufficient number of infected individuals are present to sustain reproduction. In a spatially 
homogeneous setting, this Allee threshold is given by the unstable fixed point $V_1^*$, the minimum infection level required for persistence. For localized field releases with compact spatial support, however, a more relevant threshold must account for both the height and spatial extent of the initial infection. Following \cite{BartonTurelli2011,qu2022modeling}, we refer to this space-dependent nonhomogeneous threshold as the \textit{critical bubble}.

The critical bubble is a dome-shaped profile, with infection highest at the release center and decaying with distance. If the initial condition exceeds this profile, the infection establishes locally, approaching the higher stable fixed point $V_2^*$, and then propagates outward as a traveling wave in 1-D. Conversely, if the initial condition falls below the profile, the infection collapses and is eventually lost. In the present model, the critical bubble corresponds to a nontrivial unstable steady state, $V_c^*(x)$, satisfying
\begin{equation}
V_c^*(x) = \int_{\mathbb R} K(x-y) f(V_c^*(y)) \mathrm{d}y\;.
\end{equation}

We numerically approximate this threshold using the bisection procedure of \cite{qu2022modeling}. Starting from a step initial condition (uniform infection within a finite spatial region), we vary the initial height at the release center. The critical height is identified as the value for which the solution neither collapses (\cref{fig:bubble}, left panel) nor establishes sustained persistence (middle panel) over a long observation window. The spatial profile observed at this boundary, after the initial transient phase, approximates the $V_c^*(x)$, critical bubble (\cref{fig:bubble}, right panel). The numerical accuracy can be improved by a finer spatial mesh and a longer observation period.

\begin{figure}[ht]
\centering
\includegraphics[width=0.32\textwidth]{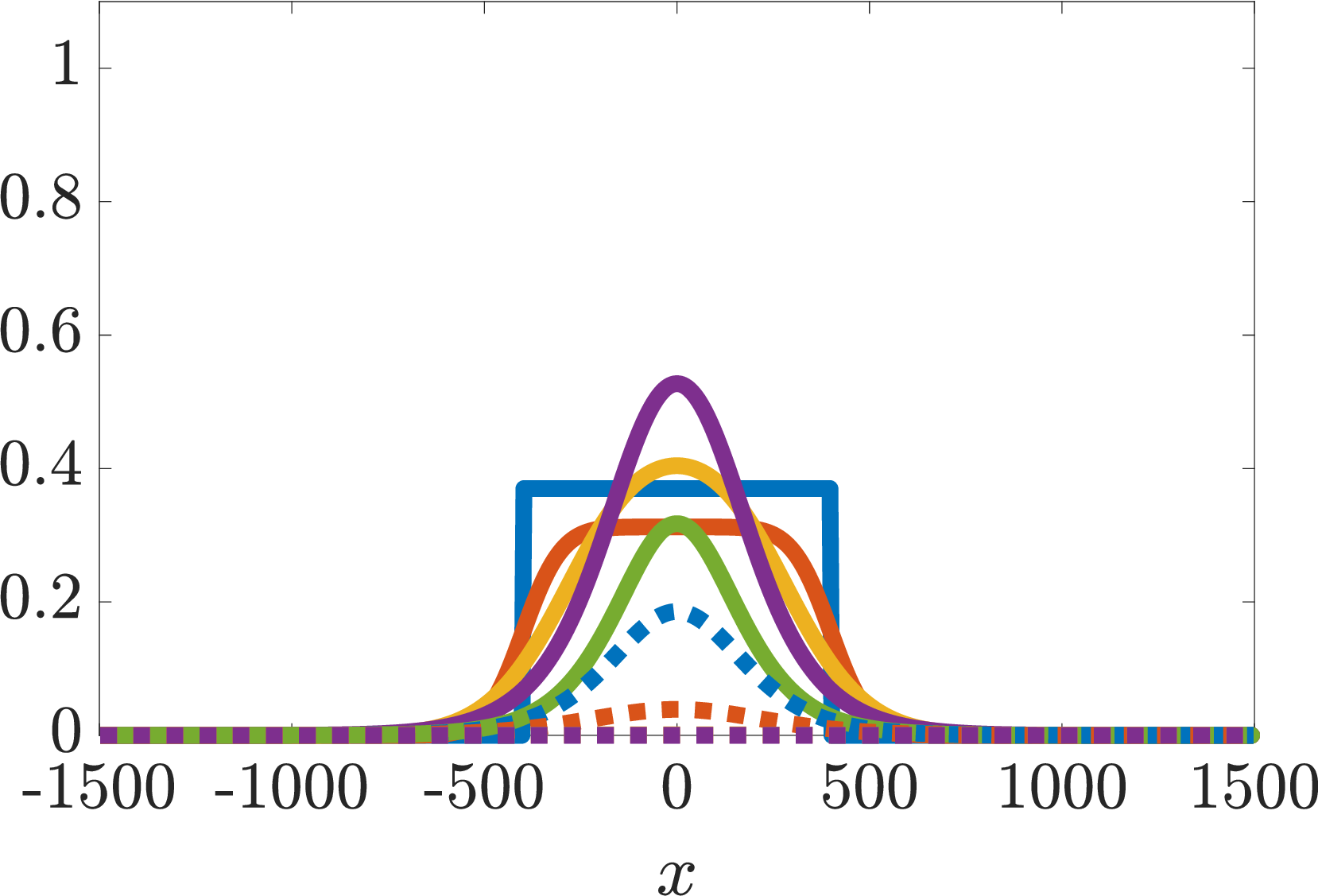}
\includegraphics[width=0.32\textwidth]{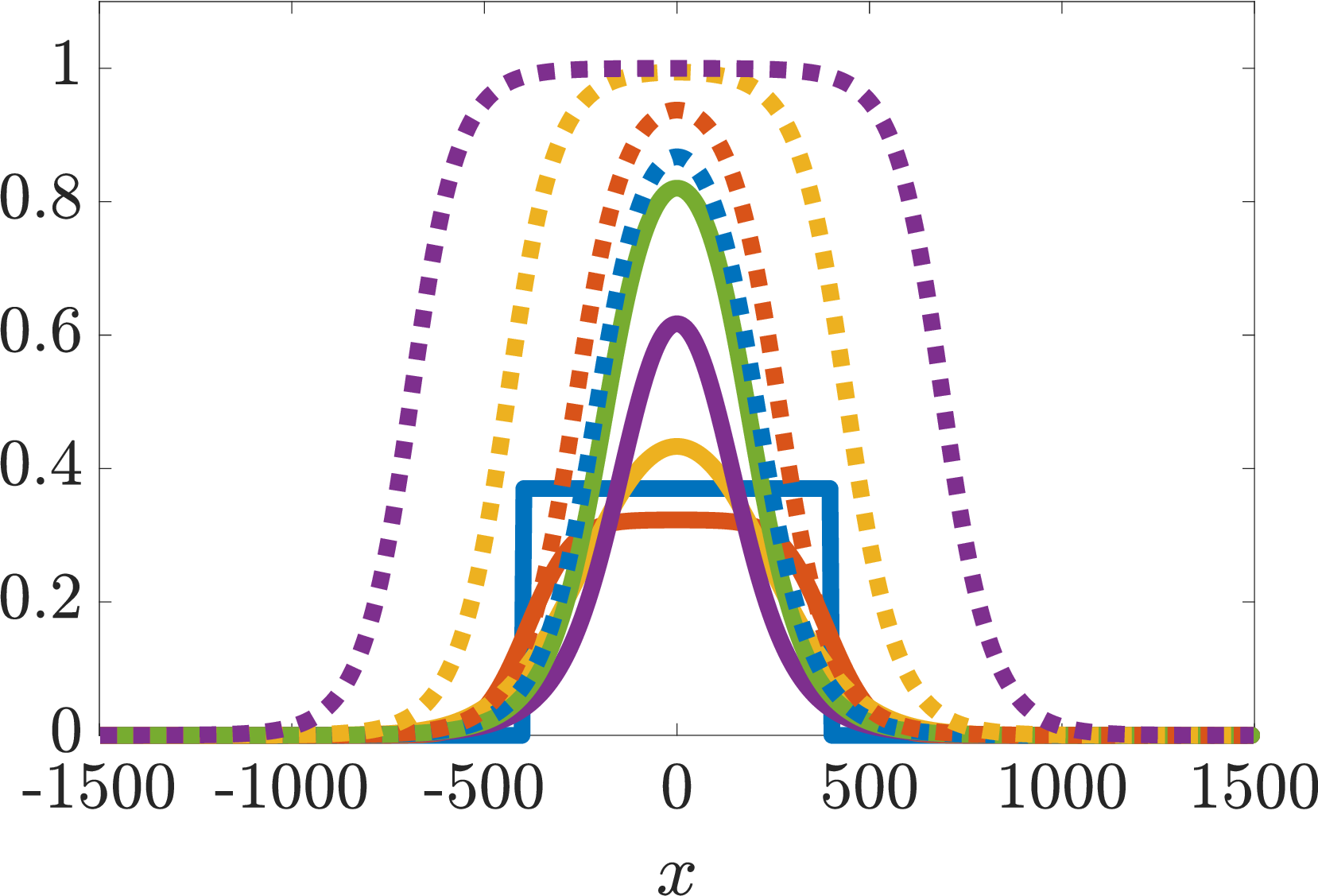}
\includegraphics[width=0.32\textwidth]{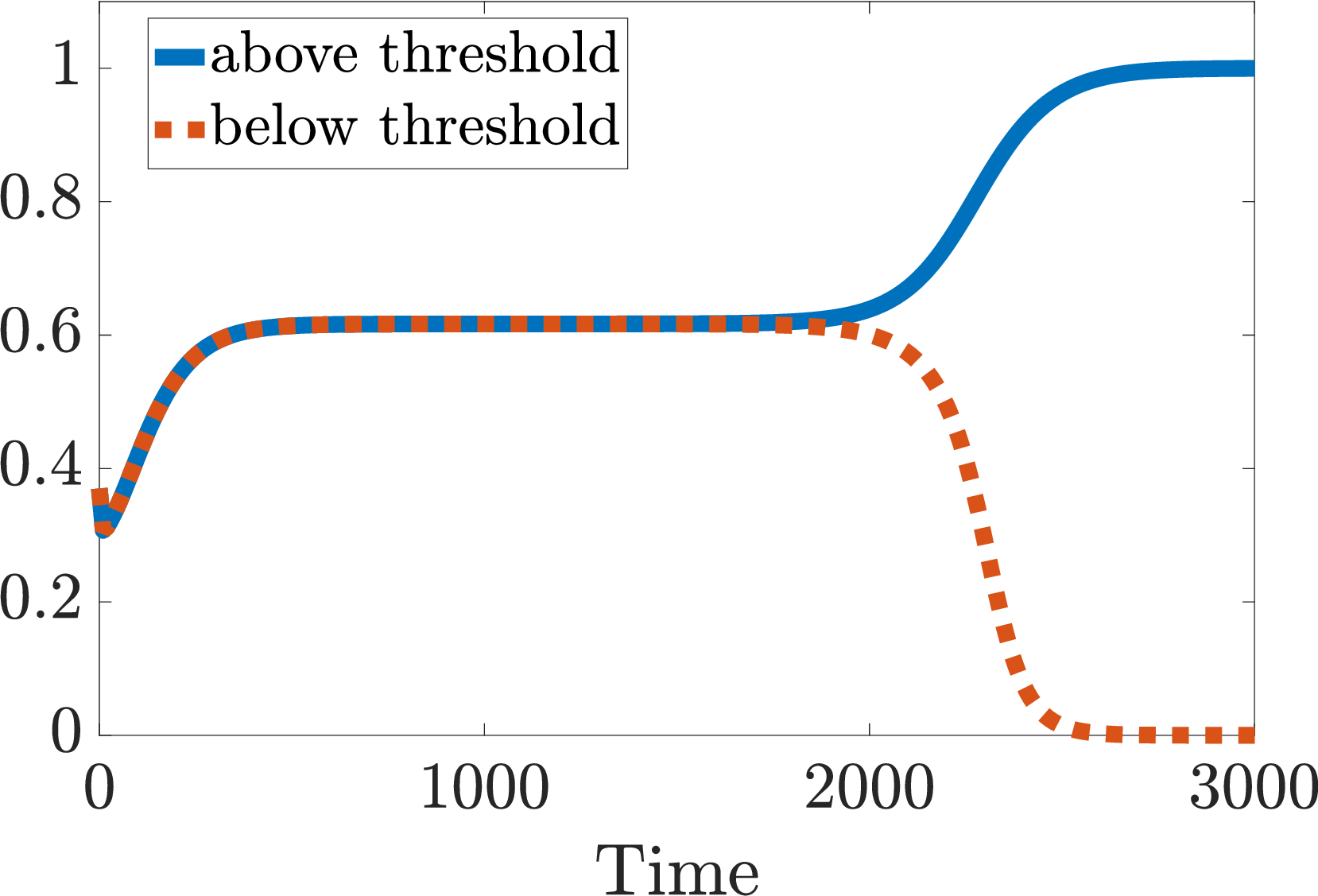}
\caption{Threshold behavior for \W ~establishment in 1-D. Left: Infection decays to extinction when the release is below the threshold. Middle: Infection establishes locally and generates an outward-propagating traveling wave when the initial release exceeds the threshold. Right: Time series of the infection fraction at the release center, $x=0$, for the left and middle panels (near the threshold), illustrating the bistable dynamics. An approximation of the critical bubble is achieved around $t\approx 1000$.}
\label{fig:bubble}
\end{figure}

\section{Numerical Results}\label{sims}
We now present numerical results that complement the analytical results (\cref{TraWaves}) on traveling waves and threshold behavior. Specifically, we implement simulations of the model using different dispersal kernels to examine the formation and propagation of traveling waves, as well as the associated threshold conditions. It is worth noting that, for heavy-tailed kernels such as the Cauchy, some of the analytical results may not strictly apply; nevertheless, we include the numerical findings here for comparison and completeness.

\subsection{Parametrization and Numerical Methods}
We consider four kernels that span short- to long-range dispersal: Gaussian, Laplace, Exponential square-root, and Cauchy. To compare the wave dynamics across kernels, we match their median absolute deviation (MAD) to a target dispersal distance $d$, which represents the lifetime distance traveled by female mosquitoes \cite{kay1998aedes,russell2005markreleaserecapture}. Kernel parameter values are summarized in \cref{tab:kernels}, and the 1-D kernel profiles are shown in \cref{fig:All_kernels} for $d = 200m$.
\begin{table}[H]
\centering
\begin{tabular}{lllll}
\toprule
kernels &  MAD (1-D) & Param. (1-D) & MAD (2-D) & Param. (2-D) \\
\midrule
Gaussian &  $ \ds \sigma \, \sqrt{2} \, \text{erf}^{-1}\left(\frac{1}{2}\right)$ & $\ds \sigma = \frac{d}{\sqrt{2}\text{erf}^{-1}\left(\frac{1}{2}\right)}$ & $\sigma \sqrt{2 \ln 2}$ & $\ds \sigma = \frac{d}{\sqrt{2 \ln 2}}$\\ 
\midrule
Laplace &$b \ln{2}$ &  $\ds b = \frac{d}{\ln{2}}$  & $\ds -b \, w$ & $\ds b = -\frac{d}{w}$\\
\midrule 
Exp. Sqrt & $\ds \frac{1}{\alpha^2}w^2 $ & 
$\ds \alpha = \frac{-w}{\sqrt{d}}$ 
& - & $\alpha \approx 0.259654 $\\ 
\midrule
Cauchy & $\beta$  & $\beta=d$ &  $\sqrt{3}\beta$ &  $\beta=d/\sqrt{3}$ \\
\bottomrule
\end{tabular}
\caption{Summary of dispersal kernels, including the median absolute deviation of the distribution (MAD) and the kernel parameter expressed in terms of the mosquito's flight range $d$. The probability density functions for both the 1-D and 2-D kernels are provided in the \cref{tab:kernels_PDFs}. $^* w=1 + W_{-1}\!\left(-\frac{1}{2e}\right)$, where $W_{-1}(\cdot)$ denotes the Lambert W function on branch $-1$. For the 2-D Exponential square-root distribution, the closed form for the MAD cannot be obtained--it satisfies the following nonlinear equation: $\ds e^{-\sqrt{d}\,\alpha} \left( 6 + \sqrt{d}\,\alpha \left( 6 + 3\sqrt{d}\,\alpha + d\alpha^2 \right) \right) = 3$. The numerical approximation is presented for $d=200$.}
\label{tab:kernels}
\end{table}

\begin{figure}[ht!] 
\centering
\includegraphics[width=0.5\textwidth]{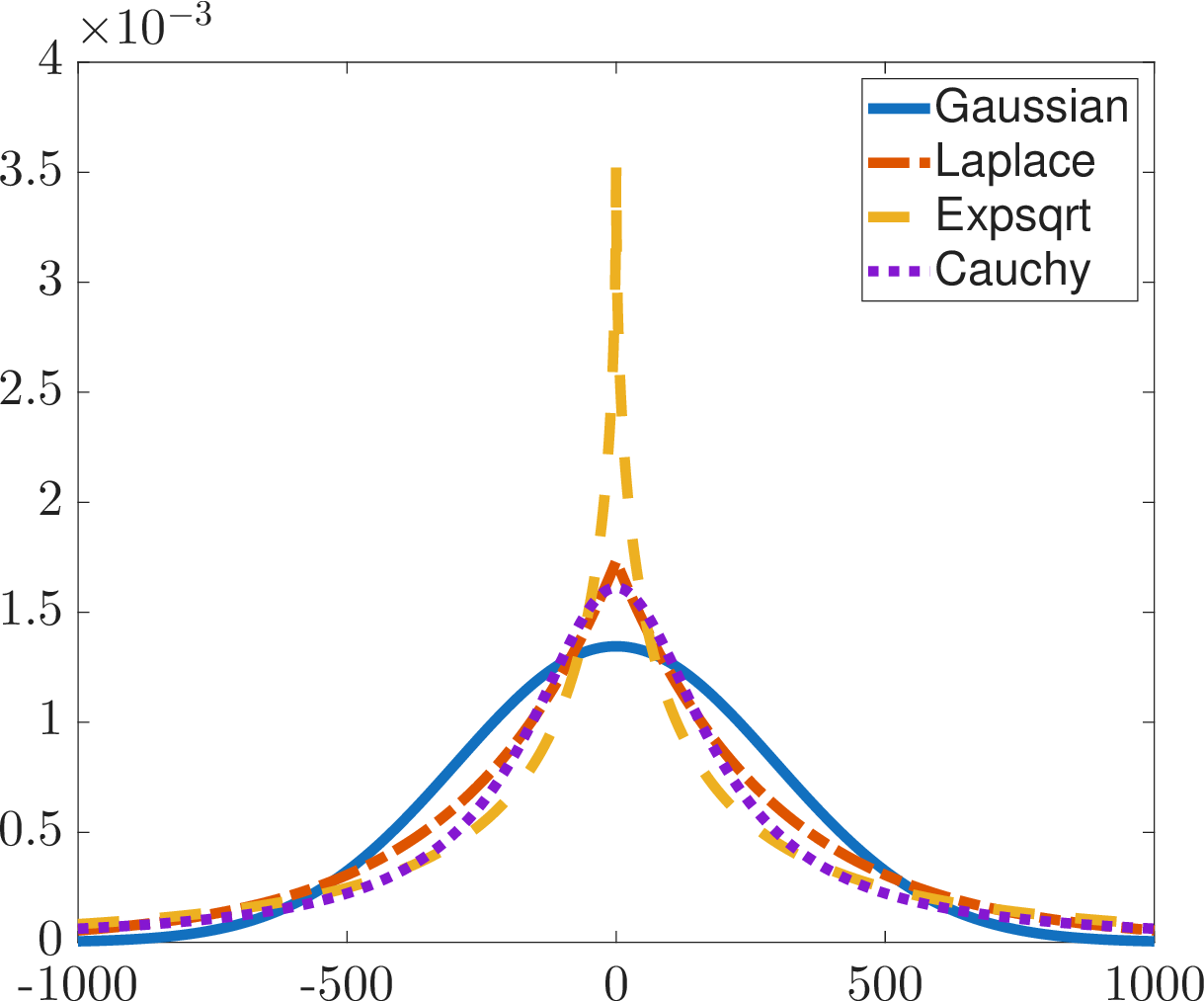} 
\caption{Dispersal kernels used in the simulations with controlled median absolute deviation $d = 200 m$. Kernel definitions and parameters are given in \cref{tab:kernels,tab:kernels_PDFs}.}
\label{fig:All_kernels}
\end{figure}

We simulate the WIDE model \cref{eq:model,eq:growth} on the spatial domain \( x \in [-L, L] \) using a spatial step size \( \Delta x \) and time step \( \Delta t = 1 \). At each time step, the convolution term in \cref{eq:model} is evaluated efficiently using the Fast Fourier Transform (FFT). Specifically, taking the Fourier transform of \cref{eq:model} yields
\begin{equation}
\mathcal{F}\{V_{t+1}\} = \mathcal{F}\left\{\int_{\mathbb R} K(x-y) f\big(V_t(y)\big) \,\mathrm{d}y\right\} 
= \mathcal{F}\{K(x)\}\,\mathcal{F}\{f(V_t(x))\},
\end{equation}
so that the convolution of the kernel \(K\) and growth term \(f(V_t)\) reduces to a simple pointwise product in Fourier space. Applying the inverse Fourier transform then gives
\begin{equation}
V_{t+1} = \mathcal{F}^{-1}\big\{\mathcal{F}\{K(x)\}\,\mathcal{F}\{f(V_t(x))\}\big\}.
\end{equation}

This spectral approach offers significant advantages over the direct numerical quadrature methods for the computation of the convolution integral. A straightforward discretization of the convolution integral requires \(O(N^2)\) operations for \(N\) spatial grid points, which becomes computationally expensive for large domains. In contrast, the FFT algorithms reduce the complexity to \(O(N\log N)\). Moreover, the FFT-based methods are more accurate for non-smooth kernels--such as the exponential square-root kernel--where standard numerical quadrature may suffer from poor convergence.

The algorithm is further generalized for a radially symmetric 2-D spatial domain, where the convolution is computed using the Hankel Transform to efficiently evaluate the radial convolution. Additional details are provided in appendix \ref{sec:2Dnum}.

\subsection{Traveling Wave Solution Profiles}\label{sec:wavesolu}
We consider a 1-D spatial domain:
$x \in [-5\times 10^4, 5\times 10^4]$ with an initial infection frequency profile given by the step function:
\begin{equation*}
V_0(x) = 
\begin{cases}
v_0, & \text{if } |x| < L \\
0, & \text{otherwise}
\end{cases}\;,
\end{equation*}
where $v_0=1$, and $L=1000$. The corresponding numerical solutions are shown in \cref{fig:comp_solu_3D,fig:comp_solu_time}. After an initial transient associated with local establishment of \W, the solutions evolve into traveling waves that propagate outward with approximately constant speed (\cref{fig:comp_solu_3D}). Additionally, for the chosen parameter values, where the kernel MAD is matched, the Cauchy kernel produces the fattest wave, followed by the Exponential square-root kernel, Laplace kernel, and then the Gaussian kernel. In the snapshots of the solutions (\cref{fig:comp_solu_time}), we further noticed that the Exponential square-root wavefront is slightly faster at the beginning and is then caught up by the Cauchy wavefront around $t = 50$. The Cauchy wavefront is slightly flatter than the others. 

\begin{figure}[ht!]
\centering
\includegraphics[width=\linewidth]{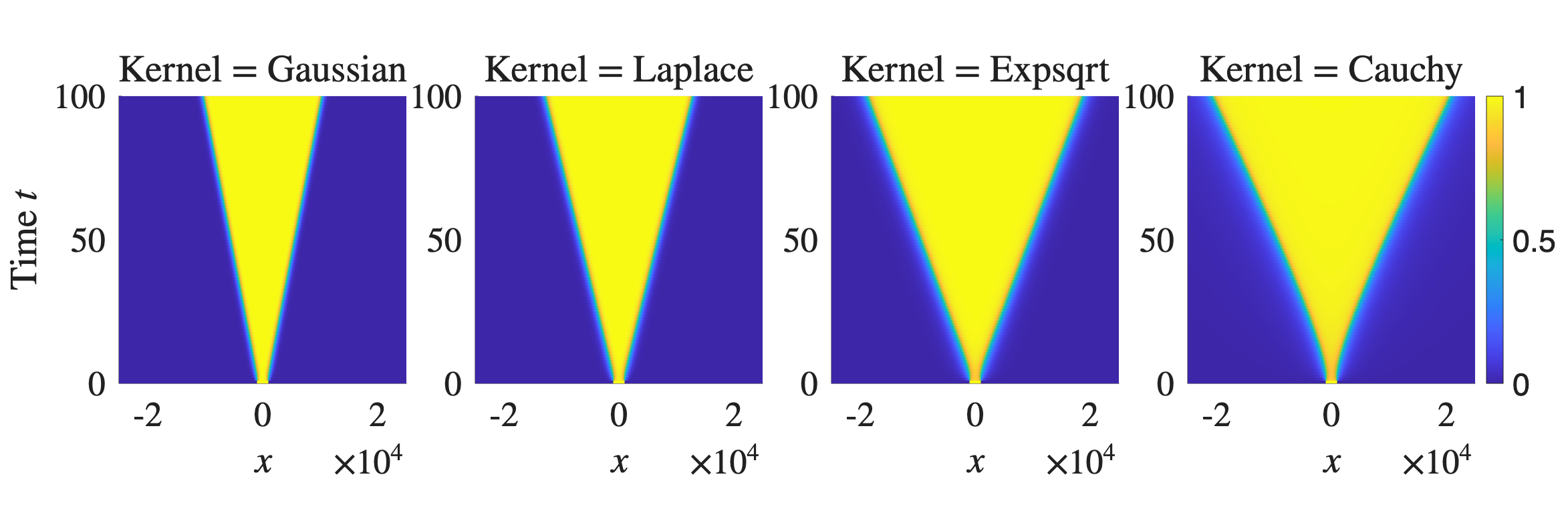}
\caption{Space–time dynamics of solutions using, from left to right, Gaussian, Laplace, Exponential square-root, and Cauchy kernels. Kernel parameters are given in \cref{tab:kernels}. Color indicates infection fraction from 0 (blue) to 1 (yellow). Traveling-wave behavior is observed for all kernels; propagation speeds at final time rank Gaussian $<$ Laplace $<$ Exponential square-root $<$ Cauchy under chosen parameter settings.}
\label{fig:comp_solu_3D}
\end{figure}

\begin{figure}[ht!]
\centering
\includegraphics[width=0.7\linewidth]{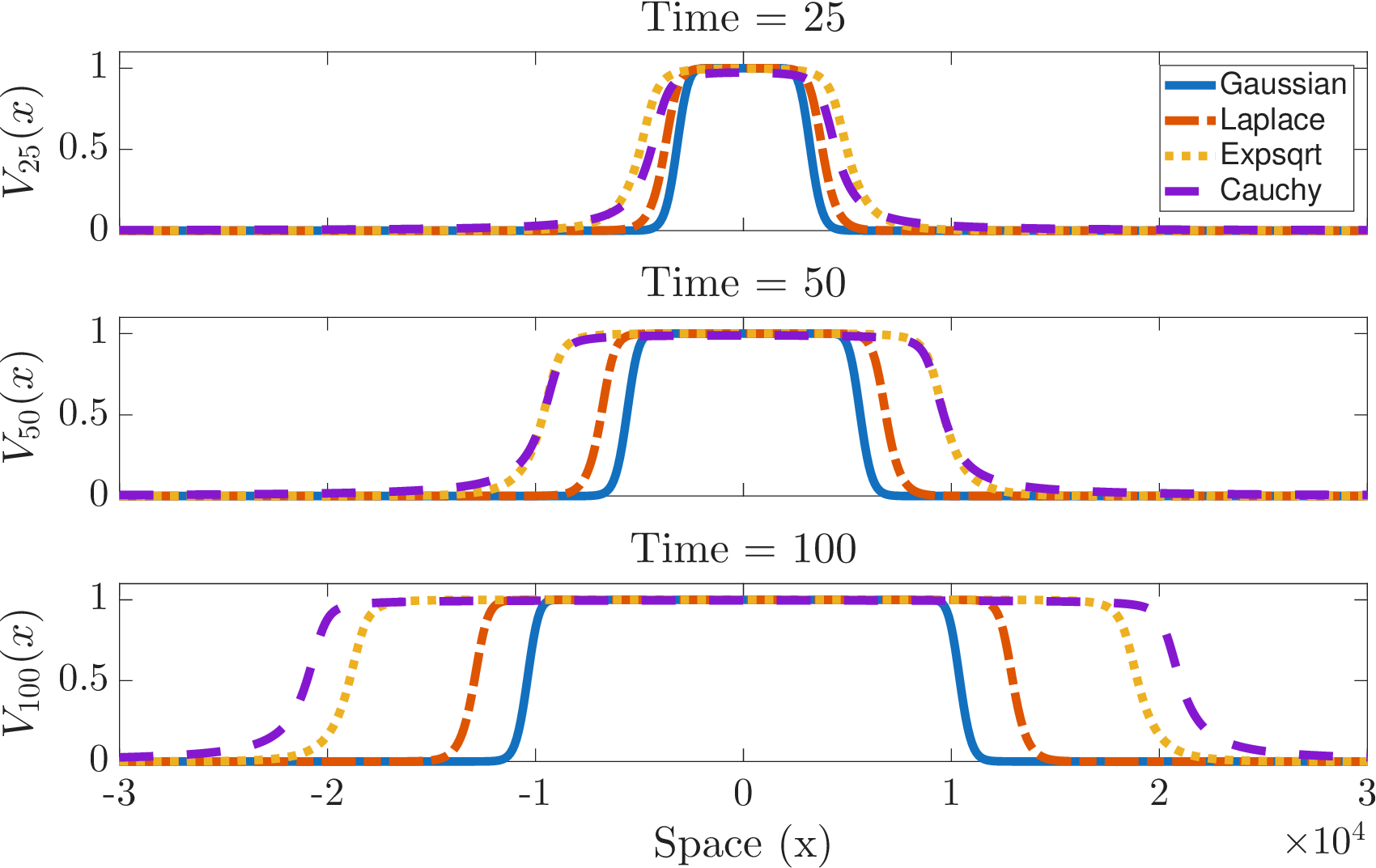}
\caption{Solution profiles at $t=25,\,50,\,100$ for all the kernels. Same numerical configurations as \cref{fig:comp_solu_3D}. The Exponential square-root wavefront is caught up by the Cauchy wavefront around $t\approx 50$. }
\label{fig:comp_solu_time} 
\end{figure}

\subsection{Numerical Estimate of Traveling Wave Speeds} \label{sec:wavespeed}
\begin{proposition}
\label{lem:lemma_Wave_speed}
Suppose the traveling  solution is  $V_t(x)=\phi(x-ct)$, where $\phi(-\infty)=V_2^*$, $\phi(\infty)=0$, and $c>0$. Then the wave speed is given as 
\[
c=\lim_{t\rightarrow\infty}\frac{1}{V_2^*}\int_0^\infty V_{t+1}(x)-V_t(x)\,dx\;,
\]
where $V_2^*$ is given in \cref{eq:V_mu}.
\end{proposition} 

\begin{proof} 
Consider $V_t(x) = \phi(x - ct)$, then $V_{t+1}(x) = \phi(x - c(t+1)) = \phi(x - ct - c).$
Substituting \( z = x - ct \), then
$$
\begin{aligned}
\int_0^\infty  V_{t+1}(x) - V_t(x) \, dx &= \int_{-ct}^{\infty} \phi(z - c) - \phi(z) \, dz = \left( \int_{-ct - c}^{\infty} \phi(\eta) \, d\eta \right) - \left( \int_{-ct}^{\infty} \phi(\eta) \, d\eta \right)\\
&= \int_{-ct - c}^{-ct} \phi(\eta) \, d\eta= c\phi(\xi),
\end{aligned}
$$
for some $\xi\in(-ct-c,-ct)$. As $t\rightarrow\infty$, we have  $\phi(\xi)\rightarrow\phi(-\infty)=V_2^*$ and 
\[
\lim_{t\rightarrow\infty}\int_0^\infty  V_{t+1}(x) - V_t(x)\, dx=cV_2^*\;.
\]
\end{proof}
\begin{remark}
In numerical simulations, we approximate the speed of the wavefront using 
$$
c(t)\approx\frac{1}{V_2^*}\int_0^\infty V_{t+1}(x)-V_t(x)\,dx,
$$
which will approach the traveling wave speed in time.
\end{remark}

The numerical wave velocities for all the kernels are shown in \cref{fig:comp_solu_velocity_MAD} (left). 
For the Gaussian, Laplace, and Exponential square-root kernels, the solutions rapidly evolve into traveling waves with nearly constant speeds of approximately
$c_{\mathrm{gaus}} \approx 96$, $c_{\mathrm{lap}} \approx 122$, and $c_{\mathrm{expsqrt}} \approx 186$, respectively. In contrast, following a longer transient phase, the velocity for the Cauchy wavefront continues to increase until about $t \approx 100$, after which it stabilizes at $c_{\mathrm{cauchy}} \approx 234$. These results confirm the propagation-speed ranking observed in the solution
profiles (\cref{sec:wavesolu}): $c_{\mathrm{gaus}} \;<\; c_{\mathrm{lap}} <\; c_{\mathrm{expsqrt}}\;<\; c_{\mathrm{cauchy}}$, when the MAD is controlled. 
\begin{figure}[H]
\centering
\includegraphics[width=0.48\textwidth]{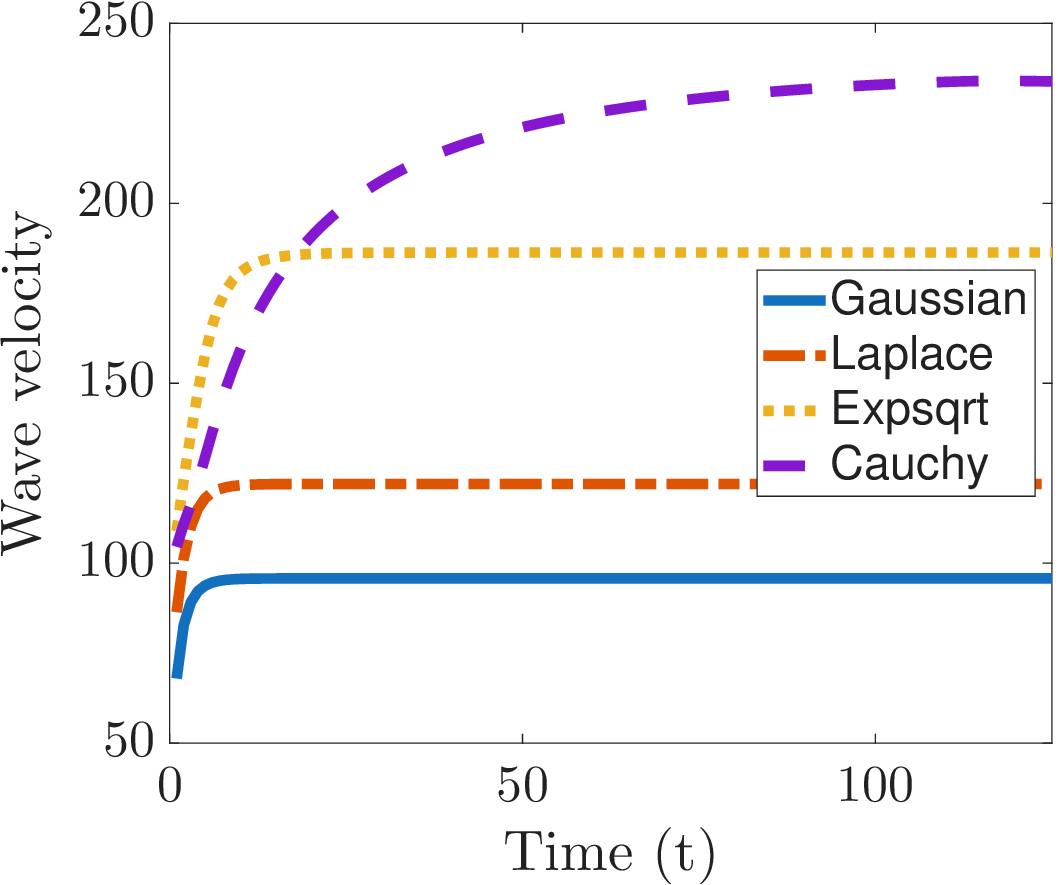}\hfill \includegraphics[width=0.48\textwidth]{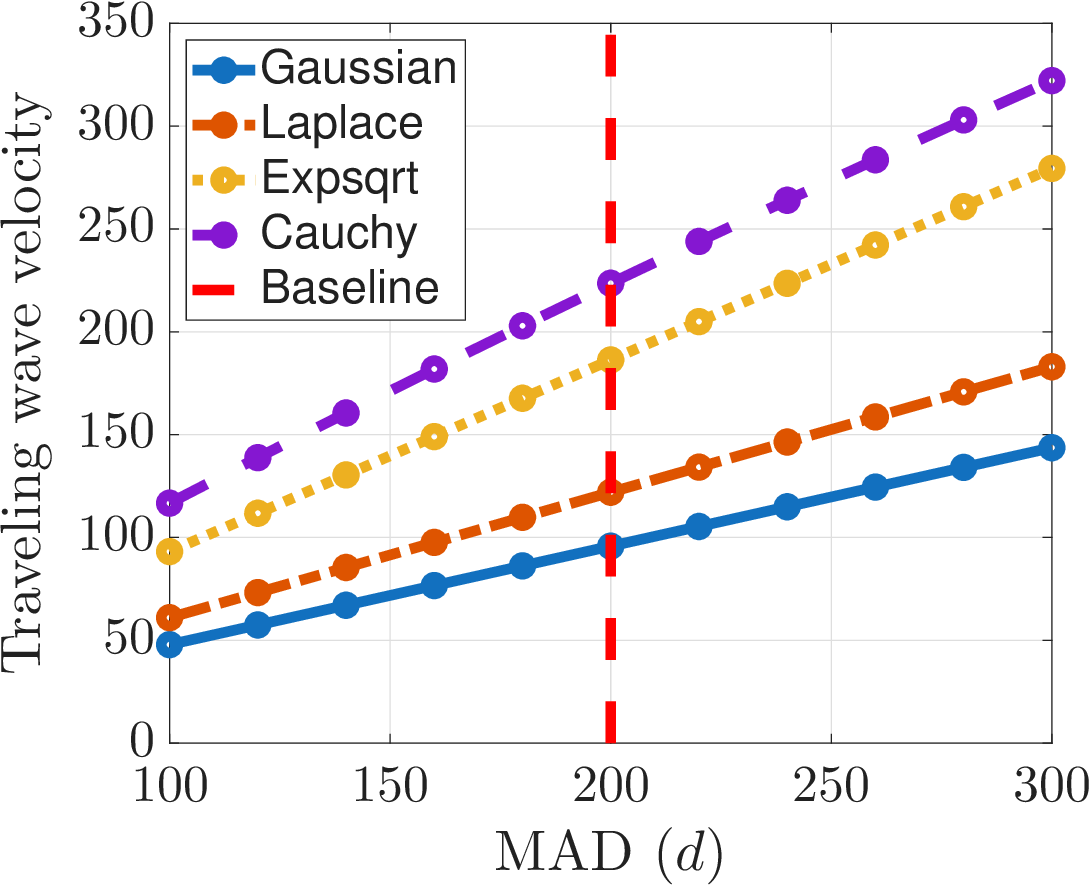}
\caption{Left: Time series of numerical wave velocities for Gaussian, Laplace, Exponential square-root, and Cauchy kernels at baseline. Right: Traveling wave velocity for varying MAD ($d$) for the four kernels. The vertical red dashed line indicates the baseline scenario seen in the left panel. The wave speed increases linearly in MAD.}
\label{fig:comp_solu_velocity_MAD} 
\end{figure}

We further numerically examine the factors that impact the wave speed and evaluate the accuracy of the analytical approximation for the sign changes in \cref{eq:sign}. Note that the simulation results below use a Riemann initial condition $V(x < 0) = V_2^*$, $V(x > 0) = 0$ for numerical efficiency.

\begin{figure}[ht!]
\centering
\includegraphics[ width=\textwidth]{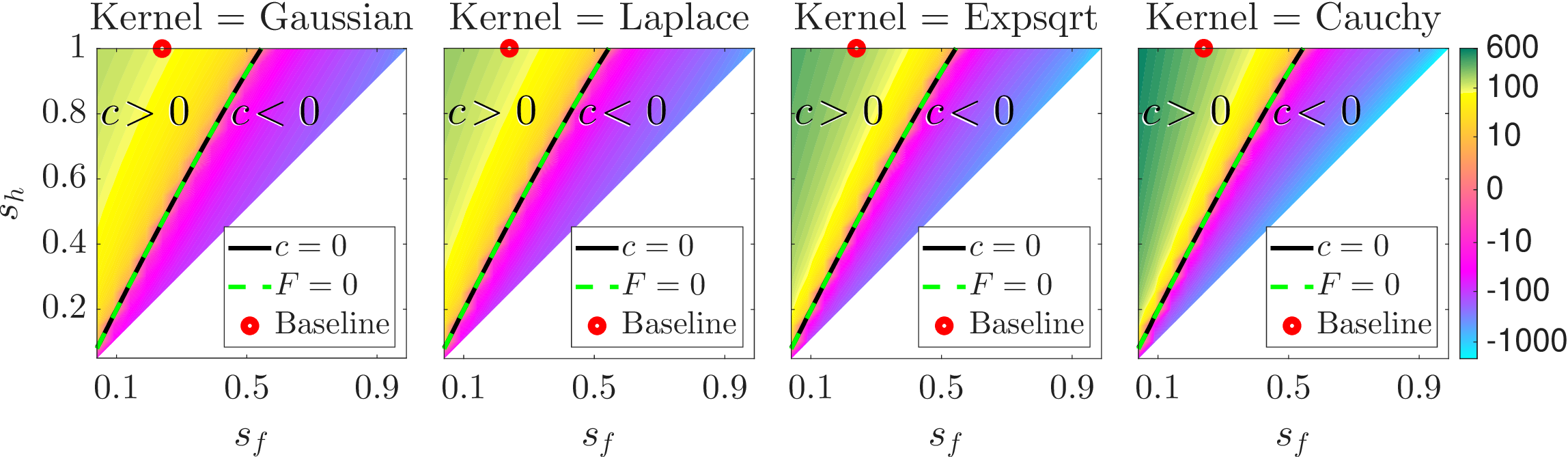}
\caption{Heatmaps of signed traveling wave speed, estimated numerically, for varying $s_f$ (fitness cost) and $s_h$ (CI intensity) for $s_f<s_h$. The solid black curve indicates $c=0$ based on the numerical estimate, and the green dashed line is the analytical prediction 
$F = \int_0^{V_2^*} (v-f(v))\,\mathrm{d}v = 0$.  The baseline scenario is marked by the red circle. The white space on the heatmap indicates the absence of the unstable fixed point due to $V_1^*=s_f/s_h>1$. }
\label{fig:heatmap}
\end{figure}

As shown in \cref{fig:heatmap}, the sign of the analytical estimate, determined by $F = \int_0^{V_2^*} (v-f(v))\,dv$ (green dashed curve), predicts accurately the direction of the propagation in the numerical simulations (black solid curve) across all four kernels. The wave switches from invasion ($c>0$, right propagation) to retreat ($c<0$, left propagation) at $c=0$, which is in close agreement with $F=0$. For all kernels, when fitness cost ($s_f$) increases, the wave velocity slows down and eventually changes its direction from invasion to retreat, making it harder for the \W~ to spread into the nearby region. Opposite trend is observed for $s_h$, that is, when CI intensity increases, the wave velocity increases, facilitating the spread of the \W~ invasion to the nearby region. These are consistent with the biological implications, where \W~ strain with lower fitness cost and high CI is desired for population replacement.
Additionally, when comparing across the four kernels with the same MAD, the fat-tailed Cauchy kernel leads to the fastest wave invasion velocity at the baseline parametrization (red circles in \cref{fig:heatmap}) and is the most sensitive to the changes in $s_f$ and $s_h$ parameters. 

When \W~parameters are fixed at baseline, the traveling wave velocity increases linearly when MAD increases for all kernels (\cref{fig:comp_solu_velocity_MAD} right). This further suggests that the estimate \cref{eq:speed} can be expressed as
$$
c \approx C\cdot \text{MAD} \cdot \int_0^1 (v-f(v))\,dv,
$$
where the coefficient $C$ is independent of the MAD of the spatial kernel.

\subsection{Comparison of Critical Bubbles} 

\Cref{fig:bubbles} presents the critical bubbles as thresholds for \W~ invasion, using the method outlined in \cref{sec:bubble}. While the overall shapes are similar, the Cauchy kernel produces the widest and tallest critical bubble, with slower spatial decay in the tails. In contrast, the Gaussian and Exponential square-root kernels produce narrower profiles with comparable peak heights and fast decay away from the center. These differences reflect the influence of kernel tail behavior: the heavy-tailed Cauchy kernel induces greater long-range dispersal (or larger diffusion), requiring a higher local infection frequency to sustain establishment of infection and thus generating a broader, taller threshold profile.

\begin{figure}[ht!]
\centering
\includegraphics[width=0.5\textwidth]{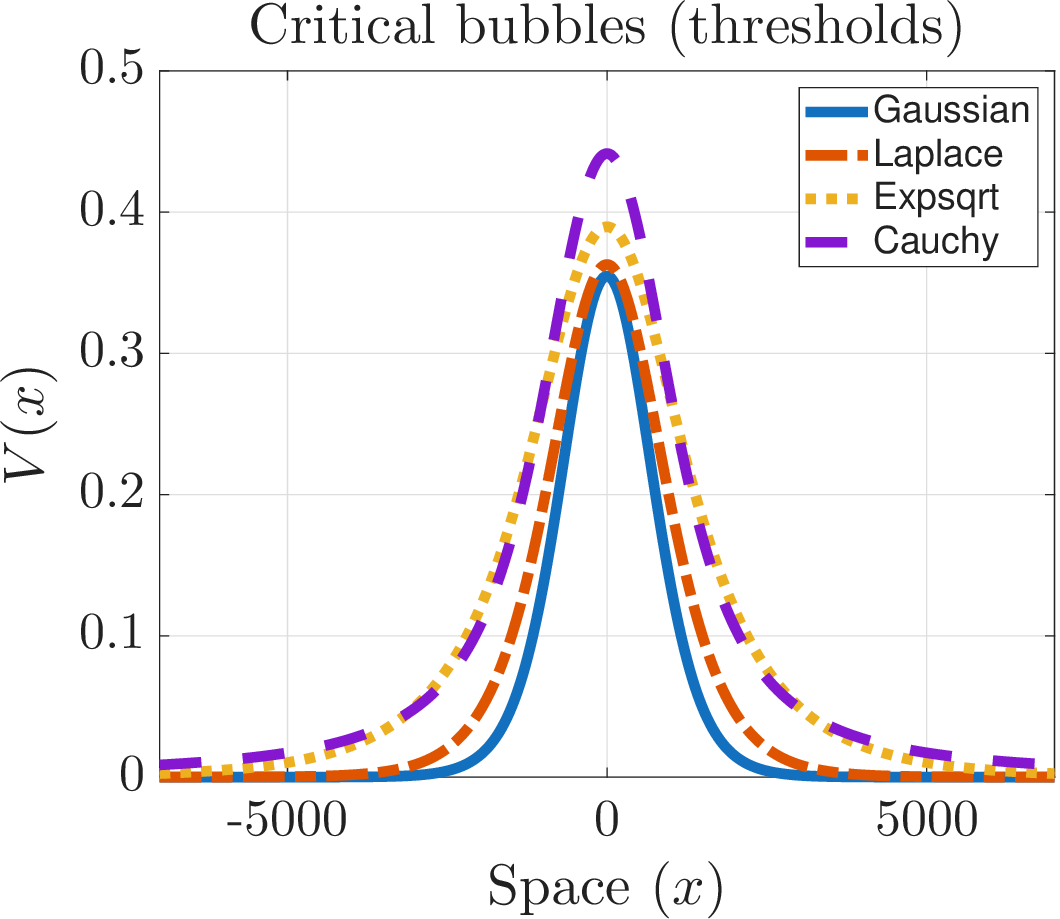} 
\caption{Critical bubbles for Gaussian, Laplace, Exponential square-root, and Cauchy kernels. Solution profiles are generated from step initial conditions set at the respective threshold heights and the critical bubbles are obtained at $t = 100$.}\label{fig:bubbles}
\end{figure}

\subsection{Results for 2-D Spatial Domain}
To investigate the spatial propagation of \W~ in a more realistic 2-D spatial domain with radial symmetry, we consider the WIDE in polar coordinates. Let $V_t(r)$ denote the infection frequency at time $t$ and radial distance $r$ from the center of release. Under the assumption of radial symmetry and isotropic dispersal, the WIDE becomes
$$
V_{t+1}(r) = \int_0^\infty K(r- \rho) f(V_t(\rho)) 2\pi \rho \, d\rho,
$$
where $K(r)$ is the radial kernel and $f(\cdot)$ is the nonlinear growth function defined in \cref{eq:growth}. The 2-D radial kernel distributions and MAD metrics are summarized in \cref{tab:kernels_PDFs} and \cref{tab:kernels}. \Cref{fig:result_2D} illustrates the numerical evolution of the infection profile under a radially symmetric initial condition \begin{equation*}
V_0(r) = 
\begin{cases}
v_0, & \text{if } r < L \\
0, & \text{otherwise}
\end{cases}\;,
\end{equation*}
where $v_0=1$, and $L=1000$. 
The solutions evolve into outward-propagating circular traveling waves after an initial transient period. Similar to the 1-D case, the Cauchy kernel produces the fastest wavefront expansion, followed by the Exponential square-root, Laplace, and Gaussian kernels. 

\begin{figure}[ht!]
\centering
\includegraphics[width=0.46\textwidth]{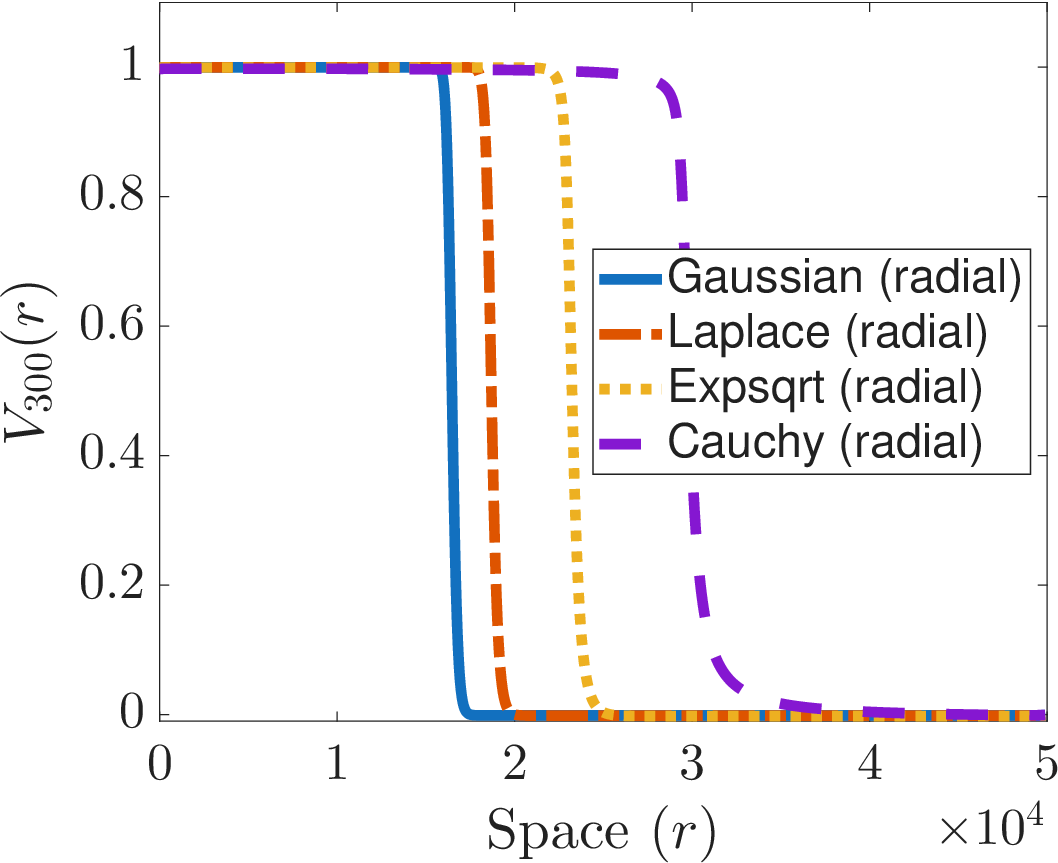}\quad  \includegraphics[width=0.48\textwidth]{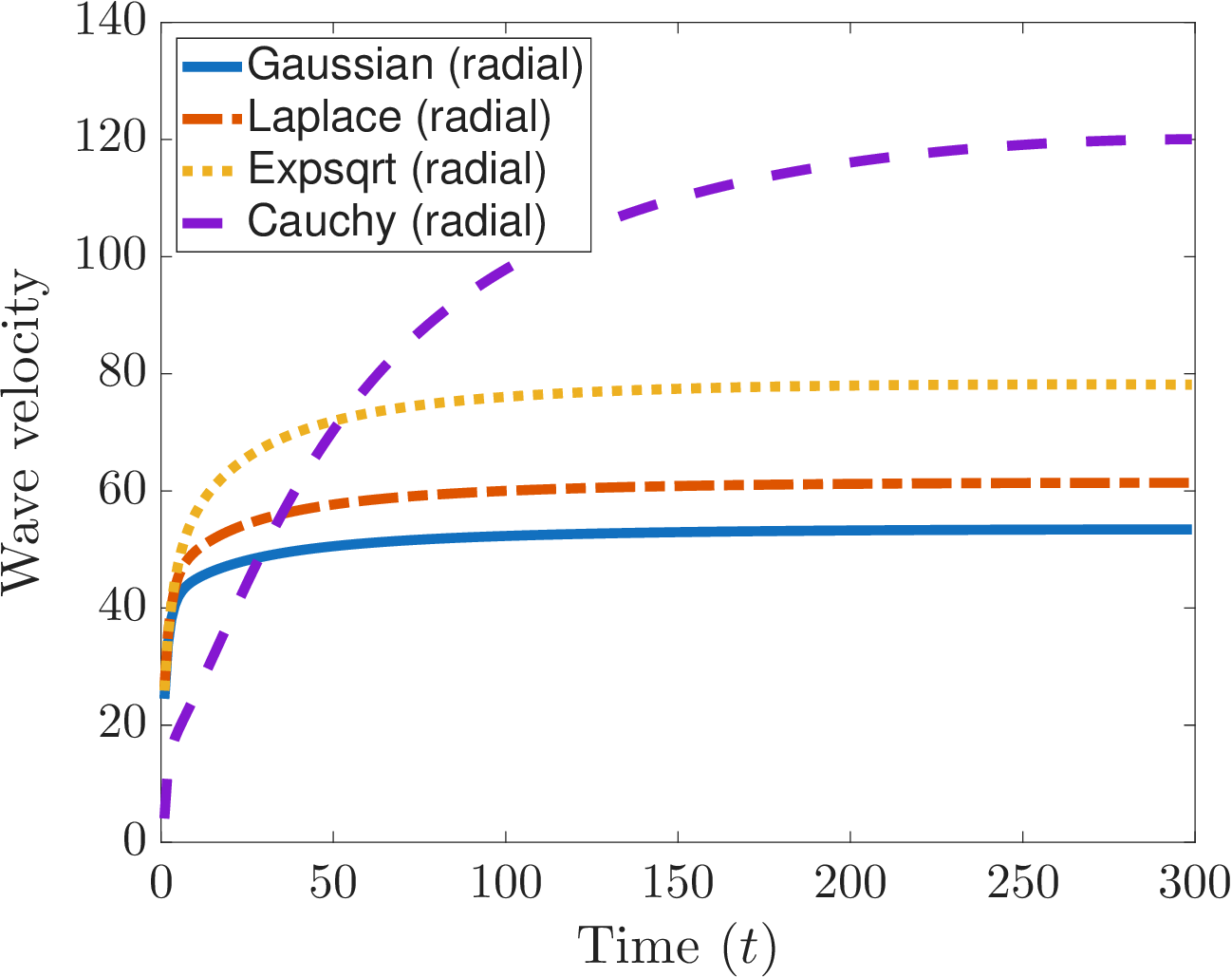}\\
\includegraphics[width=\textwidth]{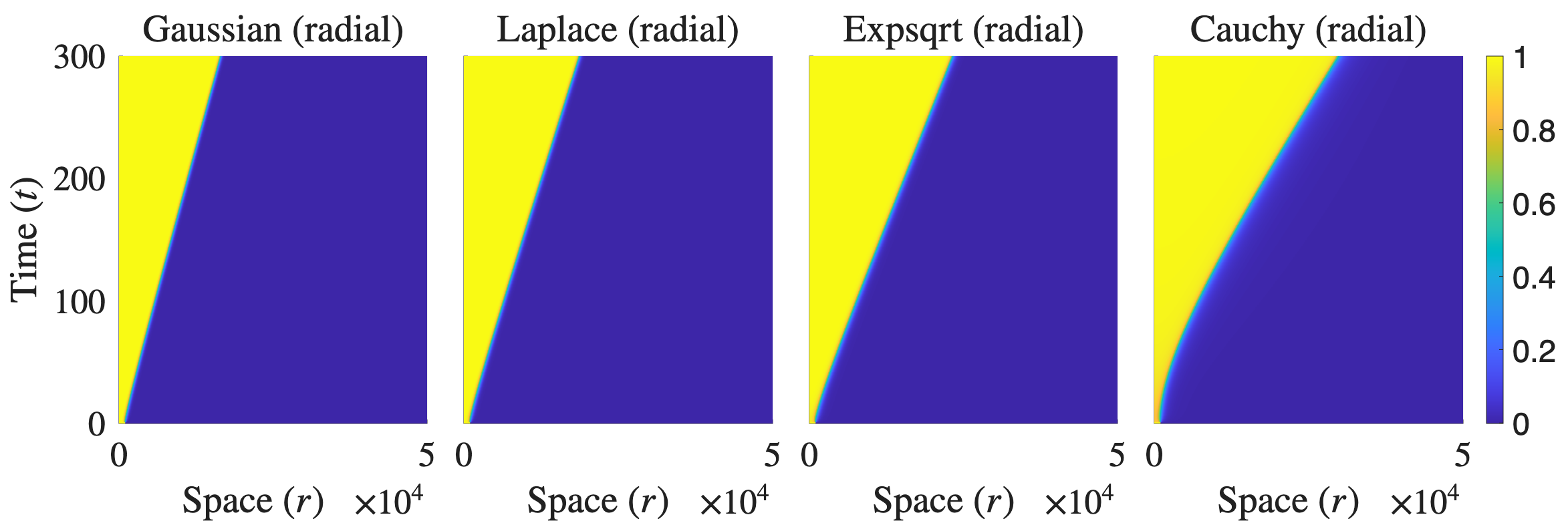}\caption{Infection dynamics under radial symmetry for different radial dispersal kernels in 2-D. Top row: Solution profile at $t=300$, and time series of corresponding wave velocities. Bottom row: Space-time heatmaps of the solutions.}
\label{fig:result_2D}
\end{figure}

\Cref{fig:results_2D_bubble} left panel depicts the critical bubble in the 2-D radially symmetric setting. Consistent with the 1-D results (\cref{fig:bubbles}), the radial Cauchy kernel yields the tallest threshold profile among the dispersal kernels considered (\cref{fig:results_2D_bubble} middle). Furthermore, the height of the critical bubble in 2-D is significantly greater than that in 1-D, assuming the MAD is controlled (\cref{fig:results_2D_bubble} right). This increase reflects the greater dispersion inherent in 2-D space domain; thus, a higher initial infection level is needed to counteract the spatial dilution and sustain the infection for invasion success.

\begin{figure}[ht!]
\centering
\includegraphics[width=0.34\textwidth]{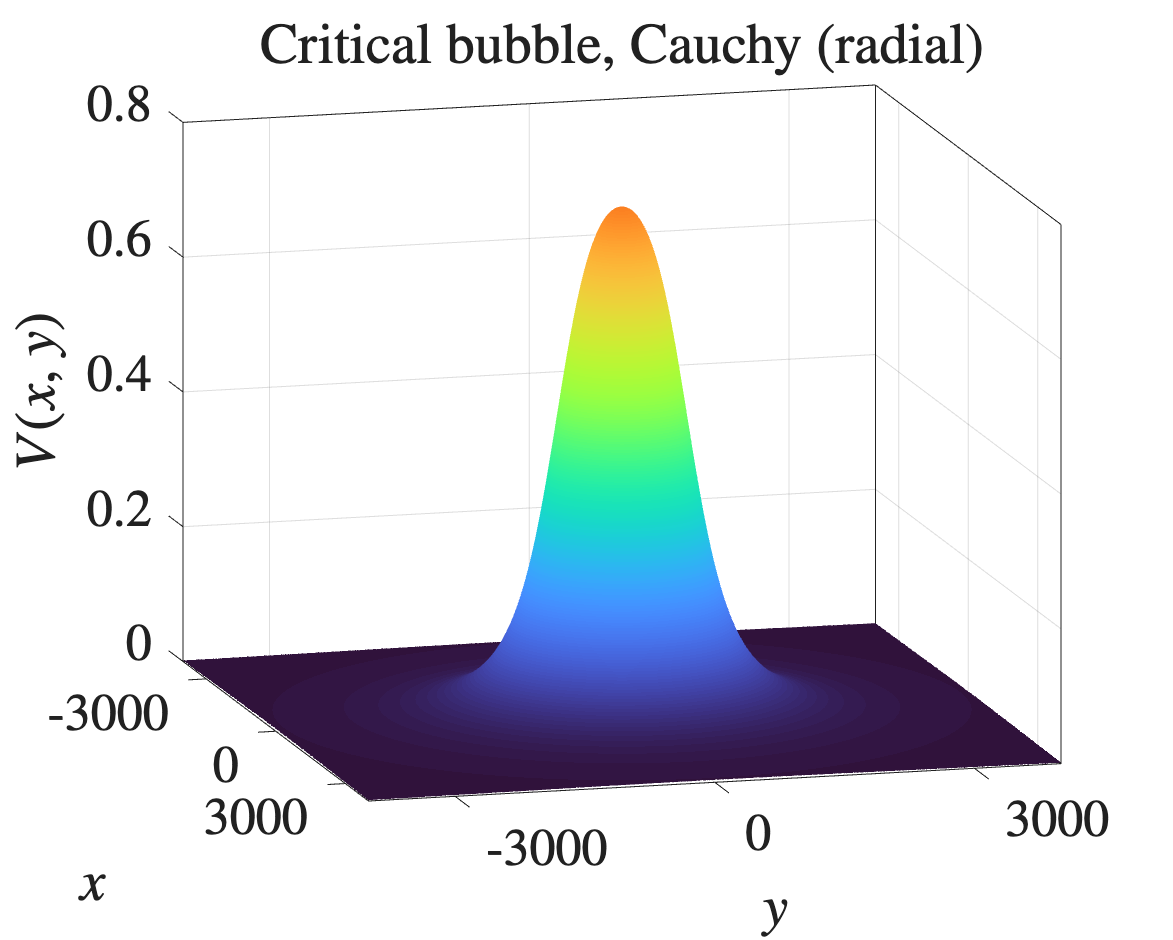}\hfill \includegraphics[width=0.32\textwidth]{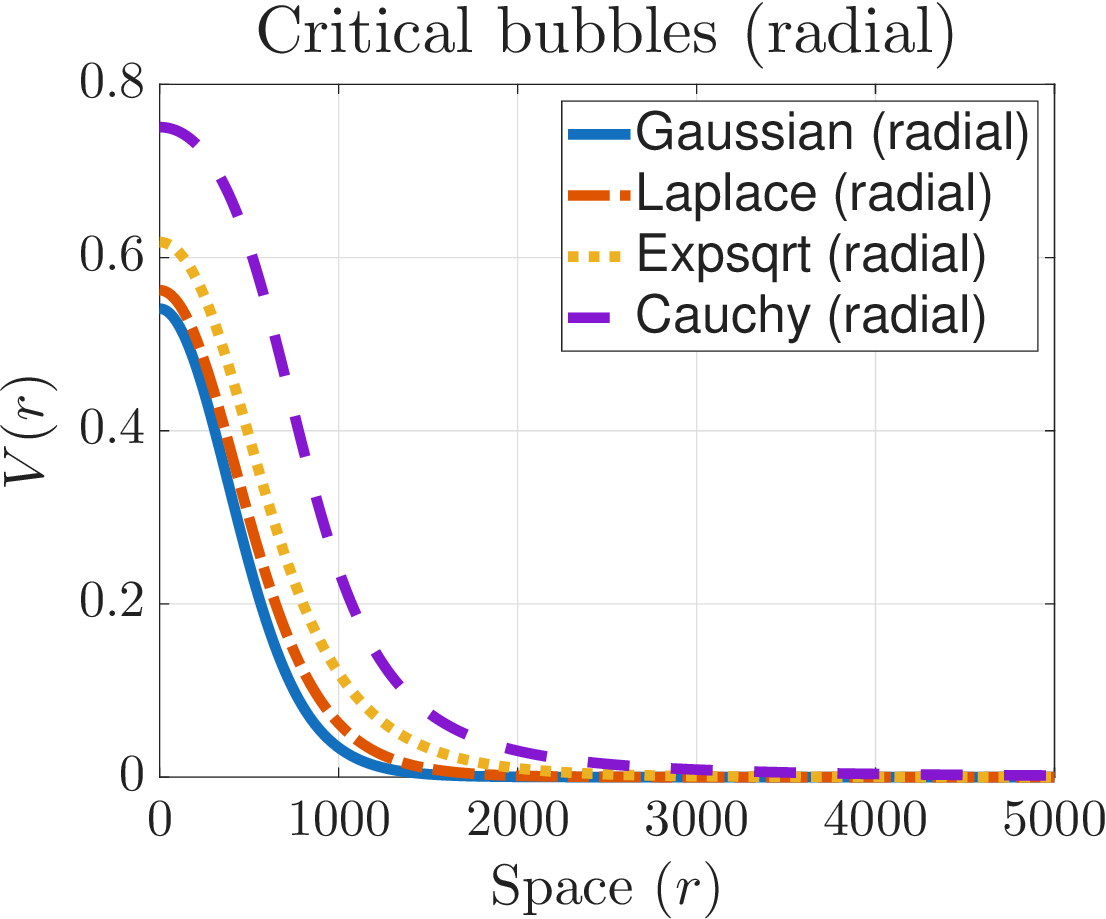}\hfill \includegraphics[width=0.32\textwidth]{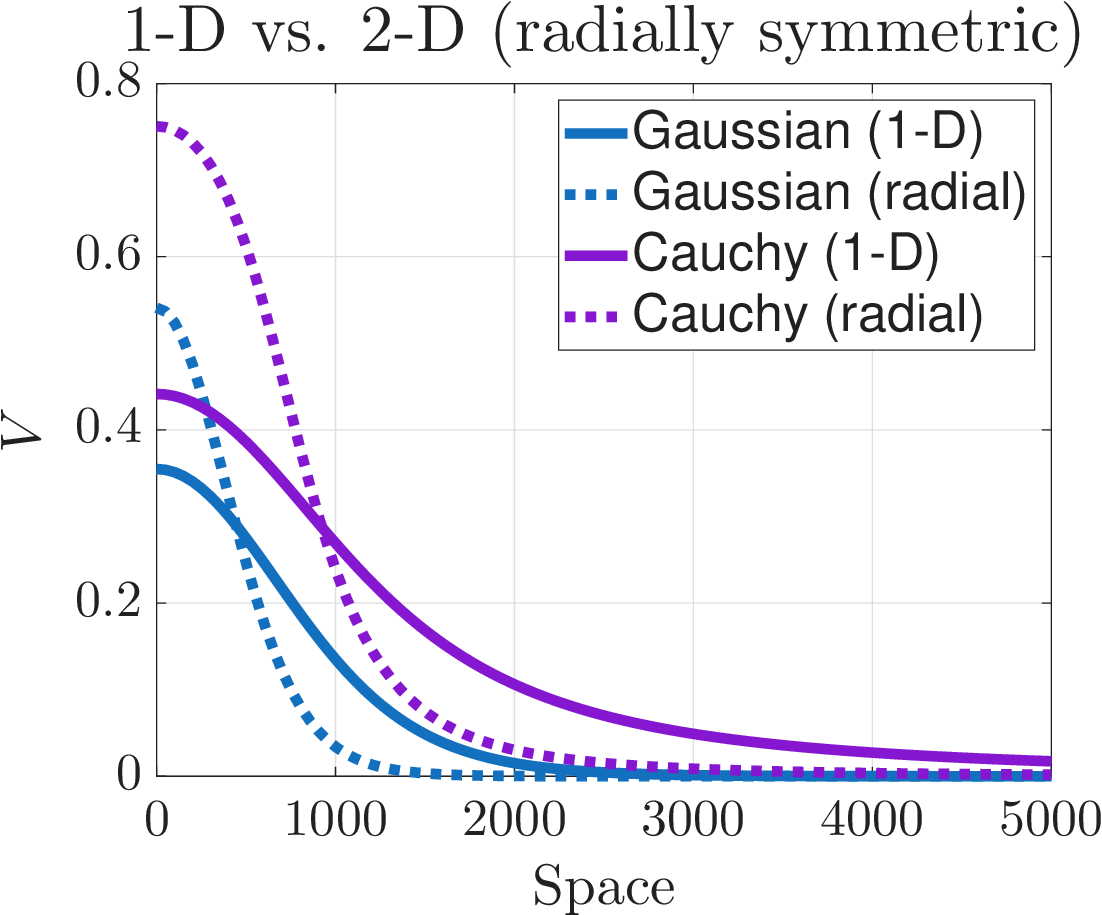}
\caption{Critical bubbles in 2-D radially symmetric setting. Left: 3-D view of the critical bubble obtained using the radially symmetric Cauchy kernel. Middle: Radial cross-sections of critical bubbles for different dispersal kernels, showing that the Cauchy kernel yields the tallest threshold. Right: Comparison between 1-D and 2-D (radially symmetric) critical bubbles under the same MAD for Gaussian and Cauchy kernels, indicating a significantly higher threshold in 2-D due to increased dispersion. }\label{fig:results_2D_bubble}
\end{figure}

\section{Discussion}\label{sec:discuss}

We have developed and analyzed an IDE framework for the spatial spread of \W~ in mosquitoes. The model couples a biologically grounded, bistable growth function with a spectrum of dispersal kernels, providing a flexible mathematical formulation that links biological mechanisms with spatial dynamics. 
Our analysis established the existence (and uniqueness up to translation) of the monotone traveling wave solutions under biologically relevant assumptions, and we derived the corresponding wave speed. We further characterized the threshold for local establishment under compact releases, represented by the critical bubble. Through a combination of analytical and numerical approaches, we demonstrated how different dispersal kernels shape the dynamics of \W~  spread. 

By systematically comparing kernels with distinct analytical properties—Gaussian, Laplace, Exponential square-root, and Cauchy—we quantified how tail heaviness influences the invasion dynamics. With the MAD controlled across the kernels, the fat-tailed Cauchy kernel produced the fastest traveling wave speed, whereas the Gaussian kernel yielded the slowest propagation. Our results also showed that the wave speed is linearly proportional to the MAD of the dispersal kernels. Moreover, the Cauchy kernel leads to the highest threshold for \W~ local establishment, implying a trade-off between the initial local reproduction and spatial spread.

The IDE framework developed here provides a theoretical baseline for understanding the spatial dynamics of \W~ infection and offers insights for field implementation. In particular, our work highlights three central factors that govern the success of \W ~ invasion (\cref{tab:wolbachia_summary}). (1) The Allee effect, captured by the critical bubble, imposes a threshold condition for local establishment with compact releases.  This underscores the importance of sufficiently high release densities to exceed the invasion threshold. (2) The dispersal kernel describes the mosquito movement patterns, which are influenced by landscape structure, and affects both the establishment threshold and the spatial spread pattern once established. The Gaussian (thin-tailed) kernel is representative of urban environments where barriers constrain mosquito flights, whereas the Cauchy (fat-tailed) kernel is more suitable for modeling open landscapes where long-range dispersal occurs. This distinction underscores the importance of context-specific modeling when applying spatial models to inform release strategies. (3) Finally, the wave speed estimate characterizes the spatial invasion rate once the local infection is established and provides a benchmark for planning release intervals, monitoring efforts, and potentially coordinating with other vector control programs. Together, these findings offer theoretical guidance for designing and optimizing \textit{Wolbachia} release programs under heterogeneous ecological conditions.

\begin{table}[H]
\centering
\begin{tabular}{|p{3.5cm}| >{\raggedright\arraybackslash}p{5.5cm}| >{\raggedright\arraybackslash}p{6cm}|}
\hline
\textbf{Key Factor} & \textbf{Implication for Dynamics} & \textbf{Field Relevance} \\
\hline
\textbf{Allee effect}\newline (Critical bubble) & \W~ cannot establish below a threshold density. & High-density initial releases are needed to surpass the invasion threshold. \\
\hline
\textbf{Dispersal kernel} & Dispersal shapes the threshold for \W~ establishment and spatial spread pattern. & Reflects the landscape heterogeneity: Gaussian fits urban settings, while fat-tailed kernels suit open landscapes. \\
\hline
\textbf{Wave speed} & Determine the rate of spatial spread once \W~ is established locally. &  Faster waves lead to quicker coverage, but with less control over spread. \\
\hline
\end{tabular}
\caption{Key model feature and properties of \W~ traveling wave dynamics, their role in spreading dynamics, and practical relevance for field implementation.}
\label{tab:wolbachia_summary}
\end{table}

While the present framework captures the essential mechanisms of the spatial dynamics, several extensions merit future investigation. The current formulation reduces mosquito dynamics to a single state variable representing \W~frequency; incorporating more detailed mosquito ecology with multiple life stages or density-dependent population dynamics could improve biological realism. Additional extensions could address heterogeneous environments, such as spatially varying parameters or fragmented landscapes, using patch-based or network-based domains to better reflect ecological complexity. Moreover, incorporating stochasticity in dispersal and infection could further capture real-world variability, while coupling with epidemiological models would enable explicit quantification of the effects of \W~ spread on disease transmission. These extensions would strengthen the connection between theoretical predictions and applied vector control efforts, reinforcing the role of spatially explicit mathematical models in the design of sustainable interventions.

\section*{Acknowledgments}
ZQ was partially supported by the National Science Foundation award DMS-2316242. The funder had no role in study design, data collection and analysis, decision to publish, or manuscript preparation.

\section*{Disclosure statement}
No potential conflict of interest was reported by the authors.

\bibliography{Spatial_Wolbachia}

\newpage 

\appendix
\section*{Appendix}
\renewcommand{\thesection}{\Alph{section}}
\crefname{section}{Appendix}{Appendices} 
\setcounter{figure}{0} 
\setcounter{equation}{0} 
\setcounter{table}{0}

\renewcommand\thefigure{\thesection.\arabic{figure}}   
\renewcommand\theequation{\thesection.\arabic{equation}}  
\renewcommand\thetable{\thesection.\arabic{table}}  

\section{Proofs for Existence of Traveling wave \label{sec:proof_lemma}}
\subsection{Proof of Lemma \ref{lem:mono}}

\begin{proof}
\noindent{(i)}  
Since $U \ge V$ and $f$ is nondecreasing, we have $f(U) \ge f(V)$.  
Because $K \ge 0$, it follows that
\[
T_c(U) - T_c(V)
= \int_{\mathbb{R}} K(z + c - y)\, \big(f(U(y)) - f(V(y))\big)\, dy
\ge 0.
\]
Hence $T_c(U) \ge T_c(V)$.

\smallskip
\noindent{(ii)}  
Let $U \in A$.  
Since $f:[0,1]\to[0,1]$ is nondecreasing, the composition $f(U(y))$ is nonincreasing and satisfies $0 \le f(U(y)) \le 1$ for all $y \in \mathbb{R}$.  
Using $K \ge 0$ and $\int_{\mathbb{R}} K = 1$, we obtain
\[
0
\le \int_{\mathbb{R}} K(z + c - y)\, f(U(y))\, dy
\le \int_{\mathbb{R}} K(z + c - y)\, dy
= 1.
\]
Thus $T_c(U)(z) \in [0,1]$.  

\noindent To verify monotonicity in $z$, let $z_1 \ge z_2$.  
By the change of variables $\xi = z_1 + c - y$, we can write
\[
T_c(U)(z_j)
= \int_{\mathbb{R}} K(\xi)\, f\big(U(z_j + c - \xi)\big)\, d\xi,
\quad j = 1,2.
\]
Since $U$ is nonincreasing and $z_1 \ge z_2$, we have
$U(z_1 + c - \xi) \le U(z_2 + c - \xi)$ for all $\xi$.  
Because $f$ is nondecreasing, this implies
$f(U(z_1 + c - \xi)) \le f(U(z_2 + c - \xi))$.  
Multiplying by $K(\xi) \ge 0$ and integrating gives
\[
T_c(U)(z_1) - T_c(U)(z_2)
= \int_{\mathbb{R}} K(\xi)
\big[f(U(z_1 + c - \xi)) - f(U(z_2 + c - \xi))\big]\, d\xi
\le 0.
\]
Hence $T_c(U)$ is nonincreasing in $z$, and therefore $T_c(U) \in A$.

\smallskip
\noindent{(iii)}  
Translation identity. We have that 
\[
T_c(U)(z) = T_0(U)(z + c).
\]
Since part (ii) shows that $T_0(U)$ is nonincreasing in its argument, for $c_1 \le c_2$ we have
\[
T_{c_1}(U)(z)
= T_0(U)(z + c_1)
\ge T_0(U)(z + c_2)
= T_{c_2}(U)(z),
\]
which establishes the claim.
\end{proof}

\subsection{Proof of Lemma \ref{lem:tails}}
\begin{proof}
Due to assumptions (K) and (A), we have 
$|K(y)f(W(z-y))|\le\|f\|_{\infty}K(y)$, with $\|f\|_{\infty}<\infty$. Moreover, $\ds\lim_{z\rightarrow\infty}K(y)f(W(z-y))=K(y)f(b)$ for all $y$, because $\lim_{z\rightarrow\infty}W(z-y)= b$ and $f$ is continuous. By the Dominated Convergence Theorem, 
$$\begin{aligned}
\lim_{z\rightarrow\infty}T_0(W)(z)&=\lim_{z\rightarrow\infty}\int K(z-y)f(W(y))dy=\lim_{z\rightarrow\infty}\int K(y)f(W(z-y))dy\\
&=\int K(y)f(b)dy=f(b)\int K(y)dy=f(b).
\end{aligned}$$
Similarly, when taking the limit $z\to -\infty$, $\ds\lim_{z\to-\infty}T_0(W)(z)=f(a)$.
\end{proof}

\subsection{Proof of Lemma \ref{lem:contact}}
\begin{proof}
(i) Since $U(z)$ is nonincreasing and 
$\lim_{z\to -\infty} U(z) = a > \theta$, 
there exists $z'$ such that $U(z) \ge \theta$ for all $z \le z'$. 
By Lemma~\ref{lem:tails}, we have 
\[
\lim_{z\to -\infty} T_0(U)(z) = f(a) \in (a,1),
\qquad 
\lim_{z\to \infty} T_0(U)(z) = f(0) = 0.
\]
Since both $U(z)$ and $f(U(z))$ are bounded and nonincreasing, and $K \in L^1$, the function $T_0(U)$—being a convolution of $K$ with $f(U)$—satisfies
\[
\sup_{x\in\mathbb{R}}
\big|T_0(U)(x+h) - T_0(U)(x)\big|
\le 
\|K(\cdot+h) - K\|_{L^1}\, \|f(U)\|_{\infty}
\xrightarrow[h\to 0]{} 0.
\]
Therefore, $T_0(U)(z)$ is uniformly continuous and nonincreasing in $z$. 
Hence there exists $z''$ such that $T_0(U)(z'') = \theta/2$.

\noindent Let $c' = z'' - z'$. Then
\[
T_{c'}(U)(z') 
= T_0(U)(z' + c') 
= T_0(U)(z'') 
= \theta/2 
< \theta 
\le U(z'),
\]
so $c' \notin C$. 
For any $c \ge c'$, Lemma~\ref{lem:mono} implies 
$T_c(U)(z') \le T_{c'}(U)(z') < U(z')$, 
and thus $c \notin C$. 
Hence $c'$ is an upper bound of $C$, and $c^* = \sup C$ is finite.

\smallskip
\noindent
We now prove the existence of $z^*$. 
Assume, to the contrary, that $T_{c^*}(U)(z) > U(z)$ for all $z$. 
Then 
\[
D := \inf_{z\in\mathbb{R}} \big(T_{c^*}(U)(z) - U(z)\big) > 0.
\]
Since $T_{c^*}(U)$ is uniformly continuous and nonincreasing in $z$, there exists $\Delta c > 0$ such that 
\[
0 \le T_{c^*}(U)(z) - T_{c^*}(U)(z+\Delta c) \le D/2,
\qquad \forall z \in \mathbb{R}.
\]
It follows that 
\[
T_{c^*+\Delta c}(U)(z) 
= T_{c^*}(U)(z+\Delta c) 
\ge T_{c^*}(U)(z) - D/2 
> U(z),
\]
which implies $c^* + \Delta c \in C$ with $c^* + \Delta c > c^*$—a contradiction to the definition of $c^* = \sup C$. 
Therefore, there exists $z^*$ such that 
$T_{c^*}(U)(z^*) = U(z^*)$.\\

\noindent (ii) Since $U(z)$ is nonincreasing and 
$\lim_{z\to \infty} U(z) = a < \theta$, 
there exists $z'$ such that $U(z) \le \theta$ for all $z \ge z'$. 
By Lemma~\ref{lem:tails}, 
\[
\lim_{z\to -\infty} T_0(U)(z) = f(1) = 1,
\qquad 
\lim_{z\to \infty} T_0(U)(z) = f(a) \in (0,a).
\]
Since $U(z)$ and $f(U(z)) \in A$ are bounded and nonincreasing, and $K \in L^1$, 
the function $T_0(U)$—as the convolution of $K$ with $f(U)$—satisfies
\[
\sup_{x\in\mathbb{R}}
\big|T_0(U)(x+h) - T_0(U)(x)\big|
\le 
\|K(\cdot+h) - K\|_{L^1}\, \|f(U)\|_{\infty}
\xrightarrow[h\to 0]{} 0.
\]
Therefore, $T_0(U)(z)$ is uniformly continuous and nonincreasing in $z$. 
Hence there exists $z''$ such that $T_0(U)(z'') = (1+\theta)/2$.

\noindent Let $c' = z'' - z'$. Then
\[
T_{c'}(U)(z') 
= T_0(U)(z' + c') 
= T_0(U)(z'') 
= (1+\theta)/2 
> \theta 
\ge U(z'),
\]
so $c' \notin C$. 
For any $c \le c'$, Lemma~\ref{lem:mono} implies 
$T_c(U)(z') \ge T_{c'}(U)(z') > U(z')$, 
and thus $c \notin C$. 
Hence $c'$ is a lower bound of $C$, and $c^* = \inf C$ is finite.

\smallskip
\noindent
We now prove the existence of $z^*$. 
Assume, to the contrary, that $T_{c^*}(U)(z) < U(z)$ for all $z$. 
Then 
\[
-D := \sup_{z\in\mathbb{R}} \big(T_{c^*}(U)(z) - U(z)\big) < 0.
\]
Since $T_{c^*}(U)$ is uniformly continuous and nonincreasing in $z$, there exists $\Delta c > 0$ such that 
\[
0 \le T_{c^*}(U)(z-\Delta c) - T_{c^*}(U)(z) \le D/2,
\qquad \forall z \in \mathbb{R}.
\]
It follows that 
\[
T_{c^*-\Delta c}(U)(z) 
= T_{c^*}(U)(z-\Delta c) 
\le T_{c^*}(U)(z) + D/2 
< U(z),
\]
which implies $c^* - \Delta c \in C$ with $c^* - \Delta c < c^*$—a contradiction to the definition of $c^* = \inf C$. 
Therefore, there exists $z^*$ such that 
$T_{c^*}(U)(z^*) = U(z^*)$.\\

\noindent (iii).  From the assumptions, we obtain
\[
1 = \lim_{z\to -\infty} U(z) > \lim_{z\to -\infty} V(z) > \theta,
\qquad
\theta > \lim_{z\to \infty} U(z) > \lim_{z\to \infty} V(z) = 0.
\]
Denote $V^L := \lim_{z\to -\infty} V(z)$ and $U^R := \lim_{z\to \infty} U(z)$. 
Since $U, V \in A$ are nonincreasing, we have $U \ge U^R$ and $V \le V^L$. 

\noindent Hence, there exists $L_1 \ge 0$ such that $U(z) > V^L$ for all $z \le -L_1$, 
and there exists $L_2 \ge 0$ such that $V(z) < U^R$ for all $z \ge L_2$. 
Then, for $z \le L_2$,
\[
U(z - L_1 - L_2) \ge U(-L_1) > V^L > V(z),
\]
and for $z > L_2$,
\[
U(z - L_1 - L_2) > U^R > V(z).
\]
Thus $U(z - L_1 - L_2) \ge V(z)$ for all $z$, implying $-L_1 - L_2 \in C$. 
Hence $C$ is nonempty.

Next, there exists $L_3$ such that $U(L_3) = (\theta + U^R)/2$, 
and there exists $L_4$ such that $V(L_4) = (\theta + V^L)/2$. 
Since $U(L_3) < V(L_4)$, $U(z + L_3 - L_4) \ge V(z)$ does not hold when $z=L_4$. 
Therefore $L_3 - L_4 \notin C$. 
By Lemma~\ref{lem:mono}, for any $c \ge L_3 - L_4$ we have $c \notin C$, 
so $L_3 - L_4$ is an upper bound of $C$. 
Thus $C$ is bounded above, and we can define $c^* = \sup C$. 
By continuity of the inequality in $c$, we have $U(z + c^*) \ge V(z)$ for all $z$.

\smallskip
\noindent
We now prove the existence of $z^*$ by contradiction. 
Assume that $U(z + c^*) > V(z)$ for all $z$. 
Then
\[
D := \inf_{z\in\mathbb{R}} \big( U(z + c^*) - V(z) \big) > 0.
\]
Since $U(z + c^*)$ is uniformly continuous and nonincreasing in $z$, there exists $\Delta c > 0$ such that
\[
0 \le U(z + c^*) - U(z + c^* + \Delta c) \le D/2,
\qquad \forall z \in \mathbb{R}.
\]
Hence
\[
U(z + c^* + \Delta c)
\ge U(z + c^*) - D/2
> V(z),
\]
which implies $c^* + \Delta c \in C$ with $c^* + \Delta c > c^*$, 
contradicting the definition of $c^* = \sup C$. 
Therefore, there exists $z^*$ such that 
$U(z^* + c^*) = V(z^*)$.

\end{proof}

\section{Parametrization: dispersal kernels}
The probability distribution functions (PDFs) for the dispersal kernels in 1-D and 2-D spatial domains are given below. 
\begin{table}[ht!]
\centering
\begin{tabular}{p{2.2cm}lll}
\toprule
kernels & $K(x)$ (1-D) & $K(r)$ (2-D) & Parameter\\
\midrule
Gaussian & $\ds \frac{1}{\sqrt{2\pi}\sigma} e^{-\frac{x^2}{2\sigma^2}}$ & $\ds \frac{1}{2\pi\sigma^2} e^{-\frac{r^2}{2\sigma^2}}\quad $& $\sigma>0$\\
\midrule 
Laplace & $\ds \frac{1}{2b} e^{-|x|/b}$  &  $\ds \frac{1}{2\pi b^2} e^{-r/b}$ & $b>0$ \\
\midrule 
Exp. Sqrt & $\ds \frac{\alpha^2}{4} e^{-\alpha \sqrt{|x|}} $ &$\ds \frac{\alpha^4}{24\pi} e^{-\alpha\sqrt{r}}$ & $\alpha>0 $ \\
\midrule 
Cauchy & $\ds \frac{\beta}{\pi}\frac{1}{\beta^2+x^2}$ & $\ds  \frac{\beta}{2\pi }\left(\frac{1}{\beta^2+r^2}\right)^{3/2}$ & $\beta>0$ \\
\bottomrule
\end{tabular}
\caption{Probability distribution functions for the dispersal kernels in 1-D and 2-D radially symmetric setting. Parameters are chosen to match the median absolute deviation to the mosquito flight data, \cref{tab:kernels}.}
\label{tab:kernels_PDFs}
\end{table}

\section{Numerical Methods: extension to 2-D radially symmetric case}\label{sec:2Dnum}

Below, we provide the numerical methods in the 2-D case. For linear kernels in general, we write
$$V_{t+1}(x,y)=\iint K(x-x',y-y')f(V_t(x',y'))\,dx'dy'\;.$$ Due to the relevance of field releases, we specifically consider radially symmetric kernels,
$$V_{t+1}(x,y)=\iint K\left(\sqrt{(x-x')^2+(y-y')^2}\right)f(V_t(x',y'))\,dx'dy'\;.$$
We use the Hankel transform to treat the convolution of radially symmetric kernels
$$V_{t+1}(x,y)=(K*f)(x,y)=\mathcal{H}_0^{-1}(\mathcal{H}_0[K](k)\times\mathcal{H}_0[f](k))\;.$$ The details are given in the following sections. 

\subsection{Definition of Hankel Transform}
The \textbf{Hankel transform of order 0} (which is the most common one used) is defined for a radial function $ f(r) $ by:
$$
\mathcal{H}_0[f](k) = \int_0^\infty f(r) \, J_0(kr) \, r \, dr\;,
$$
where:
\begin{itemize}
\item $ r $ is the radial coordinate (distance from the origin),
\item $ k $ is the radial frequency (analogous to wavenumber),
\item $ J_0(kr) $ is the \textbf{Bessel function of the first kind of order 0},
\item the extra factor $ r $ comes from the area element in polar coordinates ($ r\, dr\, d\theta $).
\end{itemize}

\subsection{Relationship with Fourier transform}
Recall the \textbf{2D Fourier transform} of $ f(x,y) $ is
$$
\mathcal{F}[f](k_x, k_y) = \iint_{\mathbb{R}^2} f(x,y) e^{-2\pi i (xk_x + yk_y)} \, dx \, dy\;.
$$
Suppose now that $ f(x,y) $ depends \textbf{only on} $ r = \sqrt{x^2 + y^2} $. So we can denote: 
$$x = r \cos\theta,\quad y = r \sin\theta,$$
$$k = \sqrt{k_x^2 + k_y^2},\quad k_x = k \cos\phi,\quad k_y = k \sin\phi,$$
which gives:
$$
xk_x + yk_y = rk \cos(\theta - \phi).
$$
Thus, the Fourier transform becomes:
$$
\mathcal{F}[f](k, \phi) = \int_0^\infty \int_0^{2\pi} f(r) e^{-2\pi i r k \cos(\theta-\phi)} r \, d\theta \, dr\;.
$$

\noindent The inner $\theta$-integral:
$$\begin{aligned}
I(k,r) =& \int_0^{2\pi} e^{-2\pi i r k \cos(\theta-\phi)} \, d\theta\\
=& \int_0^{2\pi} e^{-2\pi i r k \cos(\theta)} \, d\theta\\
=& 2\pi J_0(2\pi r k)\;.
\end{aligned}
$$
where $ J_0 $ is the \textbf{Bessel function of the first kind of order 0}.

\noindent Thus, the 2D Fourier transform becomes:
$$
\mathcal{F}[f](k) = 2\pi \int_0^\infty f(r) J_0(2\pi r k) r \, dr\;.
$$
which is the \textbf{Hankel transform} of order 0, up to a scaling factor involving $2\pi$.

\subsection{Numerical algorithm for Radial Convolution using Hankel Transform}

\paragraph{Step 1: Sample $ r $-space}
\begin{itemize}
\item Choose a grid of $ r $-values:$r_0, r_1, \dots, r_{N-1}$
uniformly spaced by $ \Delta r $.
\item Make sure: $r_0 = 0$, $r_{N-1}$ is large enough so that $f(r)$, $K(r) \approx 0$ at boundary.
\end{itemize}

\paragraph{Step 2: Sample $ k $-space}
\begin{itemize}
\item Choose $ k $-values: $k_0, k_1, \dots, k_{M-1}$ typically covering frequencies related to the $ r $-space resolution.
\item Usually $ N = M $ for simplicity.
\end{itemize}

\paragraph{Step 3: Forward Hankel Transforms}
\begin{itemize}
\item Compute the Hankel transforms of $ f $ and $ K $:
$$
\mathcal{H}_0[f](k_m) = \sum_{n=0}^{N-1} f(r_n) J_0(k_m r_n) r_n \Delta r,
\quad
\mathcal{H}_0[K](k_m) = \sum_{n=0}^{N-1} K(r_n) J_0(k_m r_n) r_n \Delta r
$$
for each $ k_m $.
\end{itemize}

\paragraph{Step 4: Multiply in $ k $-space}
\begin{itemize}
\item Multiply the transforms pointwise:
$$
\mathcal{H}_0[K*f](k_m) = \mathcal{H}_0[K](k_m) \times \mathcal{H}_0[f](k_m).
$$
\end{itemize}

\paragraph{Step 5: Inverse Hankel transform}
\begin{itemize}
\item Recover the convolved function $ (K*f)(r) $ by inverse Hankel transform:
$$
(K*f)(r_n) = \sum_{m=0}^{M-1} \mathcal{H}_0[K*f](k_m) J_0(k_m r_n) k_m \Delta k
$$
for each $ r_n $.
\end{itemize}
\end{document}